\documentclass[fleqn]{article}
\usepackage{stmaryrd,amssymb,amsmath,amsthm,url}
\title{Proposition Algebra with Projective Limits}

\author{
	Jan A.\ Bergstra\thanks{J.A. Bergstra acknowledges
	support from NWO (project Thread Algebra for Strategic
	Interleaving).} \and
	Alban Ponse \\
\\
  {\small
	  Section Theory of Computer Science,
	  Informatics Institute,
	  University of Amsterdam}\\
	{\small Url: \url{www.science.uva.nl/~{janb,alban}}}\\
	{\small Email: \url{{janb,alban}@science.uva.nl}
	}
}

\newcommand{\axname}[1]{\textup{\ensuremath{\mathrm{#1}}}}

\newcommand{\mf}{\ensuremath{\mathit{mf}}}
\newcommand{\ree}{\ensuremath{\textit{re{-}eval}\:^P}}
\newcommand{\cach}{\ensuremath{\mathit{caching}}}
\newcommand{\fr}{\ensuremath{\mathit{fr}}}
\newcommand{\rp}{\ensuremath{\mathit{rp}}}
\newcommand{\con}{\ensuremath{\mathit{cr}}}
\newcommand{\wm}{\ensuremath{\mathit{wm}}}
\newcommand{\mem}{\ensuremath{\mathit{mem}}}
\newcommand{\st}{\ensuremath{\mathit{st}}}
\newcommand{\Pmem}{\ensuremath{\mathit{Pmem}}}
\newcommand{\Nmem}{\ensuremath{\mathit{Nmem}}}

\newcommand{\SAT}{\ensuremath{\mathbf{SAT}}}
\newcommand{\FAS}{\ensuremath{\mathbf{FAL}}}
\newcommand{\dd}[1]{\textstyle{\frac{\partial}{\partial #1}}}
\newcommand{\BA}{\ensuremath{\mathbb{A}}}
\newcommand{\BB}{\ensuremath{\mathbb{B}}}

\newcommand{\trVal}{\ensuremath{\tr_{\RV}}}
\newcommand{\faVal}{\ensuremath{\fa_{\RV}}}
\newcommand{\ReVal}{\ensuremath{\Sigma_{\textit{Val}}}}
\newcommand{\RV}{\ensuremath{\mathit{Val}}}
\newcommand{\RVA}{\ensuremath{\mathrm{VA}}}
\newcommand{\andthen}{\ensuremath{\circ}}

\newcommand{\step}[1]{\ensuremath{\stackrel{#1}{\rightarrow}}}
   
\newcommand{\leftand}{~
     \mathbin{\setlength{\unitlength}{1ex}
     \begin{picture}(1.4,1.8)(-.3,0)
     \put(-.6,0){$\wedge$}
     \put(-.53,-0.36){\circle{0.6}}
     \end{picture}
     }}

\newcommand{\rightand}{~
     \mathbin{\setlength{\unitlength}{1ex}
     \begin{picture}(1.4,1.8)(0,0)
     \put(-.8,0){$\wedge$}
     \put(.72,-0.36){\circle{0.6}}
     \end{picture}
     }}

\newcommand{\leftor}{~
     \mathbin{\setlength{\unitlength}{1ex}
     \begin{picture}(1.4,1.8)(-.3,0)
     \put(-.6,0){$\vee$}
     \put(-.53,1.7){\circle{0.6}}
     \end{picture}
     }}

\newcommand{\rightor}{~
     \mathbin{\setlength{\unitlength}{1ex}
     \begin{picture}(1.4,1.8)
     \put(-.8,0){$\vee$}
     \put(.72,1.7){\circle{0.6}}
     \end{picture}
     }}

 \newcommand{\leftimp}{~
     \mathbin{\setlength{\unitlength}{1ex}
     \begin{picture}(1.5,1.8)
     \put(-.1,0){$\rightarrow$}
     \put(-.3,0.57){\circle{0.6}}
     \end{picture}
     ~}}
     
\newcommand{\rightimp}{~
     \mathbin{\setlength{\unitlength}{1ex}
     \begin{picture}(1.5,1.8)
     \put(-.9,0){$\rightarrow$}
     \put(1.76,0.57){\circle{0.6}}
     \end{picture}
     ~}}
     
\newcommand{\leftbiimp}{~
     \mathbin{\setlength{\unitlength}{1ex}
     \begin{picture}(1.5,1.8)
     \put(-.1,0){$\leftrightarrow$}
     \put(-.44,0.57){\circle{0.6}}
     \end{picture}
     ~}}
     
\newcommand{\rightbiimp}{~
     \mathbin{\setlength{\unitlength}{1ex}
     \begin{picture}(1.5,1.8)
     \put(-.9,0){$\leftrightarrow$}
     \put(1.76,0.57){\circle{0.6}}
     \end{picture}
     ~}}

\newcommand{\tr}{\ensuremath{T}}
\newcommand{\fa}{\ensuremath{F}}

\newcommand{\lef}{\ensuremath{\triangleleft}}
\newcommand{\rig}{\ensuremath{\triangleright}}

\newtheorem{theorem}{Theorem}[section]
\newtheorem{lemma}[theorem]{Lemma}
\newtheorem{proposition}[theorem]{Proposition}

\newtheorem{definition}[theorem]{Definition}  

\theoremstyle{definition}

\newtheorem*{ack}{Acknowledgment}

\newcommand{\PR}{\axname{PR}}

\newcommand{\CP}{\axname{\rm CP}}

\newcommand{\ACP}{\axname{ACP}}

\newcommand{\Nplus}{\ensuremath{\mathbb N}^+}
\newcommand{\Nat}{\ensuremath{\mathbb N}}
\newcommand{\PLM}{\ensuremath{\mathbb A}^\infty}
\newcommand{\Prop}{\ensuremath{\mathcal P}}
\newcommand{\Propinfty}{\ensuremath{\Prop^\infty}}

\RequirePackage[hyperindex]{hyperref}
\begin{document}
\date{ }
\maketitle

\begin{abstract}
Sequential propositional logic deviates from ordinary propositional 
logic by taking into account that during the sequential evaluation 
of a propositional statement,atomic propositions may yield 
different Boolean values at repeated occurrences. 
We introduce `free valuations' to capture this dynamics of a 
propositional statement's environment. 
The resulting logic is phrased as an equationally specified algebra 
rather than in the form of proof rules, and is 
named `proposition algebra'. 
It is strictly more general than Boolean algebra to the extent that 
the classical connectives fail to be expressively complete in the 
sequential case.
The four axioms for free valuation congruence are then combined with 
other axioms in order define a few more valuation congruences that 
gradually identify more propositional statements, up to static valuation 
congruence (which is the setting of conventional propositional logic).

Proposition algebra is developed in a fashion similar to the process 
algebra ACP and the program algebra PGA, via an algebraic specification 
which has a meaningful initial algebra for which a range of coarser 
congruences are considered important as well. In addition infinite 
objects (that is propositional statements, processes and programs 
respectively) are dealt with by means of an inverse limit construction 
which allows the transfer of knowledge concerning finite objects 
to facts about infinite ones while reducing all facts about 
infinite objects to an infinity of facts about finite ones in return.
\end{abstract}

\section{Introduction}
\label{sec:1}
A propositional statement is a composition 
of atomic propositions made by 
means of one or more (proposition) composition mechanisms, 
usually called connectives. 
Atomic propositions are considered to represent facts about an 
environment (execution environment, execution architecture, 
operating context) that are used by the logical mechanism 
contained in the propositional statement 
which aggregates these facts for 
presentation to the propositional statement's user. Different occurrences 
of the same atomic propositions represent different queries 
(measurements, issued information requests) at different moments 
in time. 

A valuation that may return different Boolean values for the same 
atomic proposition during the sequential evaluation of a single 
propositional statement is called \emph{free}, or in 
the case the
evaluation result of an atomic proposition can have
effect on subsequent evaluation, it is called
\emph{reactive}. 
This is in contrast to a ``static'' valuation, which always 
returns the same value for the same atomic proposition. 
Free valuations are thus semantically at the opposite 
end of static valuations, and are \emph{observation based}
in the sense that they capture the identity of a 
propositional statement as 
a series of queries followed by a Boolean value.

Many classes of valuations can be distinguished. 
Given a class $K$ of valuations, two propositional statements 
are $K$-equivalent if 
they evaluate to the same Boolean value for each valuation 
in $K$. 
Given a family of proposition connectives, $K$-equivalence 
need not be a congruence, and $K$-congruence is the largest 
congruence that is contained in $K$-equivalence. 
It is obvious that with larger $K$
more propositional statements can be distinguished 
and the one we consider most 
distinguishing is named \emph{free valuation congruence}.
It is this congruence that plays the role of an initial 
algebra for the proposition algebras developed in this paper.
The axioms of proposition algebra specify free
valuation congruence in terms of 
the single ternary connective \emph{conditional composition}
(in computer science terminology: \emph{if-then-else})
and constants for truth and falsity, and 
their soundness and completeness (for closed equations)
is easily shown. 
Additional axioms are given for static valuation congruence, and 
for some reactive valuation congruences in between.

Sequential versions of the well-known binary connectives of 
propositional logic and negation can be expressed with 
conditional composition. We prove that these connectives 
have insufficient expressive power at this level of generality 
and that a ternary connective is needed
(in fact this holds for any collection of binary connectives 
definable by conditional composition.)  

Infinite
propositional statements are defined by means of an 
inverse limit construction 
which allows the transfer of knowledge concerning finite objects 
to facts about infinite ones while reducing all facts about 
infinite objects to an infinity of facts about finite ones 
in return. This construction was applied in giving standard semantics 
for the process algebra \ACP\ (see~\cite{BK84} and for a more recent 
overview~\cite{BW90}).
In doing so the design of proposition algebra is very similar to the 
thread algebra of~\cite{BM07} which is based on a similar ternary 
connective but features constants for termination and deadlock rather 
than for truth and falsity. Whereas thread algebra focuses on 
multi-threading and concurrency, proposition algebra has a 
focus on sequential mechanisms. 

The paper is structured as follows: In the next section we discuss
some motivation for proposition algebra. 
In Section~\ref{sec:apa}
we define the signature and equations of proposition algebra, 
and in Section~\ref{sec:reval} we formally define 
valuation algebras. 
In Section~\ref{sec:rvv} we consider some
observation based equivalences and congruences 
generated by valuations, and in 
Sections~\ref{sec:com}-\ref{sec:stat}
we provide complete axiomatisations of these congruences.
Definable (binary) connectives are formally
introduced in Section~\ref{sec:add}.
In Section~\ref{sec:sat} we briefly consider
some complexity issues concerning satisfiability.
The expressiveness
(functional incompleteness)
of binary connectives is discussed in 
Section~\ref{sec:exp}. 
In Section~\ref{sec:proj} we introduce projection and 
projective limits for defining potentially infinite  
propositional statements. In Section~\ref{sec:recspec} we discuss 
recursive specifications of such propositional statements, and in 
Section~\ref{sec:app} we sketch an application perspective
of proposition algebra.
The paper is ended with some conclusions in Section~\ref{sec:conc}.

\begin{ack}
We thank Chris Regenboog for discussions about completeness 
results. We thank three referees
for their reviews of a prior version, which we think
improved the presentation considerably.
\end{ack}

\section{Motivation for proposition algebra} 
\label{sec:motiv}

Proposition algebra
is proposed as a preferred way 
of viewing the data type of propositional statements, 
at least in a context of 
sequential systems. Here are some arguments in favor of that thesis:

In a sequential program a test, which is a conjunction of $P$ 
and $Q$ will be evaluated in a sequential fashion, beginning with $P$ 
and not evaluating $Q$ unless the evaluation of $P$ led to a
positive outcome. 
The sequential form of evaluation takes precedence over the axioms or 
rules of conventional propositional logic or Boolean algebra. 
For instance, neither conjunction nor disjunction are commutative 
when evaluated sequentially in the presence of side-effects, errors 
or exceptions. The absence of these latter features is never claimed 
for imperative programming and thus some extension of ordinary 
two-valued logic is necessary to understand the basics of propositional 
logic as it occurs in the context of imperative programs. 
Three-, four- or more sophisticated many-valued logics may be used to 
explain the logic in this case (see, e.g., \cite{BBR95,BP98,Hahn05}),
and the non-commutative, sequential reading
of conjunction mentioned above can be traced back to
McCarthy's seminal work on computation theory~\cite{McC63},
in which a specific value for
undefinedness (e.g., a divergent computation) is considered
that in conjunction with falsity results 
the value that was evaluated first. 

Importing non-commutative conjunction
to two valued propositional logic means that 
the sequential order of events is significant, and that is what
proposition algebra is meant to specify and analyze in the first place. 
As a simple example, consider the propositional statement 
that a pedestrian evaluates just before
crossing a road with two-way traffic driving on the right:
\begin{equation}
\label{pedes}
\textit{look-left-and-check}\leftand\textit{look-right-and-check}
\leftand\textit{look-left-and-check}.
\end{equation}
Here $\leftand$ is \emph{left-sequential conjunction}, which is 
similar to conjunction but the left argument is
evaluated first and upon \fa\ (``false''), 
evaluation finishes with result \fa.
A valuation associated with this example is (should be) a
free valuation:
also in the case that the leftmost occurrence of 
\textit{look-left-and-check} evaluates
to \tr\ (``true''), its second evaluation might very well evaluate to \fa.
However, the order of events (or their amount) needs not to be 
significant in all circumstances and 
one may still wish or require that in particular cases 
conjunction is idempotent or even
commutative. A most simple example is perhaps
\[a\leftand a=a\]
with $a$ an atomic proposition, 
which is not valid in free valuation semantics 
(and neither is the falsity of $a\leftand \neg a$).
For this reason we distinguish a number of restricted forms of
reactive valuation equivalences and congruences that
validate this example or variations thereof, 
but still refine static valuation congruence. It is evident that
many more such refinements can be distinguished.

We take the point of departure that the very purpose of any 
action taken by a program under execution is to change the state 
of a system. If no change of state results with 
certainty the action can just as well be skipped. This holds for tests 
as well as for any action mainly performed because of its intended 
side-effects. 
The common intuition that the state is an external matter not 
influenced by the evaluation of a test, justifies ignoring side-effects 
of tests and for that reason it justifies an exclusive focus on static 
valuations to a large extent, thereby rendering the issue of 
reactivity pointless as well. But  there are some interesting cases 
where this intuition is not necessarily convincing. 
We mention three such issues, all of which
also support the general idea of considering
propositional statements under sequential evaluation:

\begin{enumerate}
\item It is common to accept that in a mathematical text an expression 
$1/x$ is admissible only after a test $x \neq 0$ has been performed. 
One might conceive this test as an action changing the state of mind of 
the reader thus influencing the evaluation of further assertions 
such as $x/x=1$. 

\item A well-known dogma on computer viruses introduced by Cohen 
in 1984 (in~\cite{Coh84}) states that a computer cannot decide whether or not 
a program that it is running is itself a virus.  The proof involves a 
test which is hypothetically enabled by a decision mechanism which is 
supposed to have been implemented on a modified instance of the machine 
under consideration. It seems fair to say that the property of a program 
being viral is not obviously independent of the state of the program. 
So here is a case where performing the test might (in principle at least) 
result in a different state from which the same test would lead to a 
different outcome. 
\\
This matter has been analyzed in detail in~\cite{BP05,BBP07} with the 
conclusion that the reactive nature of valuations gives room for criticism 
of Cohen's original argument. In the didactic literature on computer security 
Cohen's viewpoint is often repeated and it can be found on many 
websites and in the introduction of many texts. But there is a remarkable 
lack of secondary literature on the matter; an exception is
the discussion in~\cite{Coh01} and the papers cited therein. 
In any case the common claim that 
this issue is just like the halting problem (and even more important in 
practice) is open for discussion.
  
\item The on-line halting problem is about execution 
environments which allow a running program to acquire information about 
its future halting or divergence. 
This information is supposed to be provided by means of a forecasting 
service. 
In~\cite{BP07} that feature is analyzed in detail in a setting of 
thread algebra and the impossibility of sound and complete forecasting of 
halting is established. In particular, calling a forecasting
service may have side-effects which leads to different replies 
in future calls (see, e.g.,~\cite{PZ06}).
This puts the impossibility of effectively deciding the 
halting problem at a level comparable to the impossibility of finding a 
convincing truth assignment for the liar paradox sentence in conventional 
two valued logic.
\end{enumerate}
  
Our account of proposition algebra is based on
the ternary operator \emph{conditional composition} 
(or \emph{if-then-else}). 
This operator has a  sequential form of
 evaluation as its natural semantics,
and thus combines naturally with free and reactive valuation semantics.
Furthermore, proposition algebra
constitutes a simple setting for constructing 
infinite propositional statements by means of an inverse 
limit construction. The resulting
\emph{projective limit model} can
be judged as one that didactically precedes (prepares for)
technically more involved versions for
process algebra and thread algebra, and as such provides 
by itself a motivation for proposition algebra.

\section{Proposition algebra}
\label{sec:apa}
In this section we introduce the signature and
equational axioms of proposition algebra. Let $A$ 
be a countable set of atomic propositions
$a,b,c,...$. The elements of $A$ serve 
as atomic (i.e., non-divisible) queries that will produce a 
Boolean reply value.

We assume that $|A|>1$. The case that $|A|=1$ is described
in detail in~\cite{Chris}. We come back to this point
in Section~\ref{sec:conc}.

The signature of proposition algebra consists of the constants  
$\tr$ and $\fa$ (representing true and false), 
a constant $a$ for each $a \in A$,
and, following Hoare in \cite{Hoa85}, the ternary operator 
\emph{conditional composition} 
\[\_\lef \_ \rig \_ \;.\]
We write $\Sigma_{\CP}(A)$ for the signature introduced here. 
Terms are subject to the equational axioms in Table~\ref{CP}. 
We further write \CP\ for this set of axioms (for conditional propositions).

\begin{table}
\centering
\hrule
\begin{align*}
	\nonumber&(\axname{CP1}) &x \lef \tr \rig y &= x\\
	\nonumber&(\axname{CP2}) &x \lef \fa \rig y &= y\\
	\nonumber&(\axname{CP3}) &\tr \lef x \rig \fa  &= x\\
    \nonumber&(\axname{CP4}) &
    x \lef (y \lef z \rig u)\rig v &= 
	(x \lef y \rig v) \lef z \rig (x \lef u \rig v)
\end{align*}
\hrule
\caption{The set \CP\ of axioms for proposition algebra}
\label{CP}
\end{table}

An alternative name for the conditional composition 
$y\lef x\rig z$ is 
\emph{if $x$ then $y$ else $z$}: 
the axioms CP1 and CP2 model that its central
condition $x$ is evaluated first, and depending on the reply 
either its
leftmost or rightmost argument is evaluated. 
Axiom CP3 establishes that a term 
can be extended to a larger conditional composition by adding 
$\tr$ as a leftmost argument and $\fa$ as a 
rightmost one, and CP4 models the way
a non-atomic central condition distributes over the outer arguments.
We note that the expression 
\[\fa\lef x\rig\tr\]
can be seen as defining the negation of $x$:
\begin{equation}
\label{ss}
\CP\vdash 
z\lef(\fa\lef x\rig \tr)\rig y
=(z\lef\fa\rig y)\lef x\rig (z\lef\tr\rig y)=y\lef x\rig z
,
\end{equation}
which illustrates that
 ``\emph{if $\neg x$ then $z$ else $y$}" 
and ``\emph{if $x$ then $y$ else $z$}"
are considered equal.

We introduce the abbreviation 
\[x\andthen y\quad\text{for}\quad y\lef x\rig y,\]
and we name this expression \emph{$x$ and then $y$}.
It follows easily that $\andthen$ is associative:
\[(x\andthen y)\andthen z=z\lef(y\lef x\rig y)\rig z=
(z\lef y\rig z)\lef x\rig(z\lef y\rig z)=
x\andthen (y\andthen z).
\] 
We take the and-then operator
\andthen\ to bind stronger than conditional composition.
At a later stage we
will formally add negation, the ``and then'' connective
\andthen, and some other binary connectives to proposition
algebra (i.e., add their function symbols to $\Sigma_\CP(A)$ 
and their defining equations to \CP).

Closed terms over $\Sigma_{\CP}(A)$ are called
\emph{propositional statements}, 
with typical elements $P,Q,R,...$. 

\begin{definition}
\label{propdef}
A propositional statement $P$ is a \textbf{basic form} if 
\begin{align*}
P ::= \tr\mid\fa\mid P_1\lef a \rig P_2
\end{align*}
with $a\in A$, and $P_1$ and $P_2$ basic forms.
\end{definition}
So, basic forms can be seen as binary trees of which the leaves 
are labeled with either $\tr$ or $\fa$, and the internal nodes 
with atomic propositions. Following \cite{BW90} we
 use the name \emph{basic form}
instead of \emph{normal form} because we associate the latter
with a term rewriting setting.

\begin{lemma}
\label{proplem}
Each propositional statement can be proved equal to 
one in basic form using the axioms in Table~\ref{CP}.  
\end{lemma}

\begin{proof}
We first show that if $P,Q,R$ are basic forms, then
$P\lef Q\rig R$ can be proved equal to a basic form  by structural induction
on $Q$. If $Q=\tr$ or $Q=\fa$, this follows immediately, and
if $Q=Q_1\lef a\rig Q_2$ then 
\begin{align*}
\CP\vdash P\lef Q\rig R&=P\lef(Q_1\lef a\rig Q_2)\rig R \\
&=
(P\lef Q_1\rig R)\lef a\rig (P\lef Q_2\rig R)
\end{align*}
and by induction there are basic forms $P_i$ for $i=1,2$
such  that
$\CP\vdash P_i=P\lef Q_i\rig R$, 
hence $\CP\vdash P\lef Q\rig R= P_1\lef a\rig P_2$
and $P_1\lef a\rig P_2$ is a basic form.

Next we prove the lemma's statement by
structural induction on the form that propositional
statement $P$ may take. 
If $P=\tr$ or $P=\fa$ we are done, and if
$P=a$, then $\CP\vdash P=\tr\lef a \rig \fa$. 
\\
For the case $P=P_1\lef P_2\rig P_3$ it follows by induction that
there are basic forms $Q_1$, $Q_2$, $Q_3$ with 
$\CP\vdash P_i=Q_i$, 
so
$\CP\vdash P=Q_1\lef Q_2\rig Q_3$. We proceed by case 
distinction on $Q_2$:
if $Q_2=\tr$ or $Q_2=\fa$
the statement follows immediately; if $Q_2=R_1\lef a\rig R_2$,
then
\begin{align*}
\CP\vdash P=Q_1\lef Q_2\rig Q_3
&=Q_1\lef (R_1\lef a\rig R_2)\rig Q_3\\
&=
(Q_1\lef R_1\rig Q_3)\lef a\rig (Q_1\lef R_2 \rig Q_3)
\end{align*}
and by the first result there are basic forms $S_1,S_2$
such that $\CP\vdash S_i=Q_1\lef R_i\rig Q_3$. Hence
$\CP\vdash P=S_1\lef a\rig S_2$ and $S_1\lef a\rig S_2$ is a 
basic form.
\end{proof}

We write 
\[P\equiv Q\]
to denote that propositional statements
$P$ and $Q$ are syntactic equivalent.
In Section~\ref{sec:rvv} we prove that
basic forms constitute a convenient representation:

\begin{proposition}
\label{propo}
If $\CP\vdash P=Q$ for basic forms $P$ and $Q$, then 
$P\equiv Q$.
\end{proposition}

\section{Valuation algebras}
\label{sec:reval}
In this section we formally define valuation algebras.
Let $B$ be the sort of the Booleans with constants
\tr\ and \fa.
The signature $\ReVal(A)$ of valuation algebras 
contains the sort $B$ and 
a sort \RV\ of valuations. The sort $\RV$ has two constants
\[\trVal\quad\text{and}\text\quad\faVal,\]
which represent the valuations that assign to 
each atomic proposition the value \tr\ respectively \fa,
and for each $a \in A$ a function
\[y_a:\RV\rightarrow B\]
called the \emph{yield} of $a$, and a function
\[\dd a:\RV\rightarrow\RV\]
called the $a$-derivative. 
A $\ReVal(A)$-algebra is thus a two-sorted algebra.

A $\ReVal(A)$-algebra 
$\BA$ over $\ReVal(A)$ is a \emph{valuation algebra} 
(\RVA) if for all $a\in A$ it satisfies the axioms
\begin{align*}
y_a(\trVal)&=\tr,\\
y_a(\faVal)&=\fa,\\
\dd a(\trVal)&=\trVal,\\
\dd a(\faVal)&=\faVal.
\end{align*}

Given a valuation algebra $\BA$
with a valuation $H$ (of sort \RV)
and a propositional statement
$P$ we simultaneously define the evaluation
of $P$ over $H$, written
\[P/H,\]
and a generalized notion of an $a$-derivative for 
propositional statement $Q$, 
notation
\(\dd Q(H)\)
by the following case distinctions:
\begin{align*}
\tr/H&=\tr,\\
\fa/H&=\fa,\\
a/H&=y_a(H),\\
(P\lef Q\rig R)/H&=
\begin{cases}
P/\dd Q(H)&\text{if $Q/H =\tr$},\\
R/\dd Q(H)&\text{if $Q/H =\fa$},
\end{cases}
\end{align*}
and
\begin{align*}
\dd\tr(H)&=H,\\
\dd\fa(H)&=H,\\
\dd{(P\lef Q\rig R)}(H)&=
\begin{cases}
\dd P(\dd Q(H))&\text{if $Q/H =\tr$},\\
\dd R(\dd Q(H))&\text{if $Q/H =\fa$}.
\end{cases}
\end{align*}
Some explanation: whenever in a conditional
composition
the central condition is an atomic proposition, say $c$, then
a valuation $H$ distributes over the outer arguments as $\dd c(H)$, 
thus
\begin{align}
\label{oso}
(P\lef c\rig Q)/H=
\begin{cases}
P/\dd c(H)&\text{if $y_c(H) =\tr$},\\
Q/\dd c(H)&\text{if $y_c(H) =\fa$}.
\end{cases}
\end{align}
If in a conditional
composition the central condition is not atomic, valuation
decomposes further according to the equations above, e.g.,
\begin{align}
\nonumber
(a\lef ( b\lef c\rig d)\rig e)/H&=
\begin{cases}
a/\dd{( b\lef c\rig d)}(H)&\text{if $( b\lef c\rig d)/H =\tr$},\\
e/\dd{( b\lef c\rig d)}(H)&\text{if $( b\lef c\rig d)/H =\fa$},
\end{cases}
\\
&=
\label{isi}
\begin{cases}
a/\dd b(\dd c(H))&\text{if $y_c(H)=\tr$ and $y_b(\dd c(H))=\tr$},\\
a/\dd d(\dd c(H))&\text{if $y_c(H)=\fa$ and $y_d(\dd c(H))=\tr$},\\
e/\dd b(\dd c(H))&\text{if $y_c(H)=\tr$ and $y_b(\dd c(H))=\fa$},\\
e/\dd d(\dd c(H))&\text{if $y_c(H)=\fa$ and $y_d(\dd c(H))=\fa$}.
\end{cases}
\end{align}

We compare the last example with 
\begin{align}
\label{asa}
((a\lef b\rig e)\lef c\rig (a\lef d\rig e))/H,
\end{align}
which is a particular instance of
\eqref{oso} above.
For the case $y_c(H)=\tr$ we find from \eqref{oso} that
\begin{align*}
\eqref{asa}=(a\lef b\rig e)/\dd c(H)=
\begin{cases}
a/\dd b(\dd c(H))&\text{if $y_b(\dd c(H))=\tr$},\\
e/\dd b(\dd c(H))&\text{if $y_b(\dd c(H))=\fa$},\\
\end{cases}
\end{align*}
and for the case $y_c(H)=\fa$ we find the other two right-hand sides
of \eqref{isi}. In a similar way it follows that  
\[\dd{((a\lef b\rig e)\lef c\rig (a\lef d\rig e))}(H)
=\dd{(a\lef ( b\lef c\rig d)\rig e)}(H),\]
thus providing 
a prototypical example of the soundness of 
axiom CP4 of \CP.
Without (further) proof we state the following result.

\begin{theorem}[Soundness]
\label{thm:sound}
If for propositional statements $P$ and $Q$,
$\CP\vdash P=Q$, then for all \RVA s \BA\ and all
valuations $H\in \BA$,
$P/H=Q/H$ and $\dd P(H)=\dd Q(H)$.
\end{theorem}

\begin{proof}
Let $\BA$ be some \RVA\ and $H\in\BA$.
It is an easy exercise to show that an arbitrary instance 
$P=Q$ of
one of the axioms in $\CP$ satisfies $P/H=Q/H$ and 
$\dd P=\dd Q$.
\end{proof}

We note that $\CP\vdash P=Q~\Rightarrow~\dd P(H)=\dd Q(H)$)
ensures that
the congruence property is respected, e.g., if 
for some $H$, $\tr=P/H=Q/H$, then
\[(R\lef P\rig S)/H=R/\dd P(H)=R/\dd Q(H)=(R\lef Q\rig S)/H\]
and $\dd{(R\lef P\rig S)}(H)=\dd R(\dd P(H))=\dd R(\dd Q(H))
=\dd{(R\lef Q\rig S)}(H)$.

\section{Valuation varieties}
\label{sec:rvv}
We introduce some specific equivalences and congruences 
generated by classes of valuations.
The class of \RVA s that satisfy a certain collection of equations over 
$\ReVal(A)$ is called a \emph{valuation} variety.
We distinguish the following varieties, where each next 
one is subvariety of the one defined:

\begin{enumerate}
\item The variety of \RVA s with \emph{free}
valuations: 
no further \RVA-equations than those defined in 
Section~\ref{sec:reval}.

\item 
\label{rep}
The variety of \RVA s with \emph{repetition-proof} valuations:
all \RVA s that satisfy for all $a\in A$,
\begin{equation*}
\label{rva:rp}
y_a(x)=y_a(\dd a(x)).
\end{equation*}
So the reply to a series of consecutive atoms $a$ is 
determined by the first reply.

Typical example: $(P\lef a \rig Q)\lef a \rig (R\lef a\rig S)
=(a\andthen P)\lef a \rig (a\andthen S)$.

\item The variety of \RVA s with \emph{contractive}
valuations: 
all repetition-proof \RVA s that 
satisfy for all $a\in A$,
\begin{equation*}
\label{rva:cr}
\dd a(\dd a(x))=\dd a(x).
\end{equation*}
Each successive atom $a$ 
is contracted by using the same evaluation result.

Typical example: $(P\lef a \rig Q)\lef a \rig (R\lef a\rig S)
=P\lef a \rig S$.

\item 
\label{rva:wm}
The variety of \RVA s with \emph{weakly memorizing}
valuations consists of
all contractive \RVA s that
satisfy for all $a,b\in A$,
\begin{align*}
&y_b(\dd a(x))=y_a(x)\rightarrow[\dd a(\dd b(\dd a(x)))
=\dd b(\dd a(x))
~\wedge~
y_a(\dd b(\dd a(x)))=y_a(x))].
\end{align*}
Here the evaluation result of 
an atom $a$ is memorized in a subsequent evaluation
of $a$ if the evaluation of intermediate atoms 
yields the same result, and this subsequent $a$
can be contracted. 

Two typical examples are
\begin{align*}
((P\lef a \rig Q)\lef b \rig R)\lef a \rig S
&=(P\lef b\rig R)\lef a \rig S,\\
P\lef a\rig(Q\lef b \rig (R\lef c\rig (S\lef a\rig V)))
&=P\lef a\rig(Q\lef b \rig (R\lef c\rig V)).
\end{align*}
The case in which there are two intermediate atoms is 
discussed in Section~\ref{sec:wm}.
 
\item The variety of \RVA s with \emph{memorizing}
valuations: 
all contractive \RVA s that
satisfy for all $a,b\in A$,
\begin{align*}
&\dd a(\dd b(\dd a(x)))
=\dd b(\dd a(x))
~\wedge~
y_a(\dd b(\dd a(x)))=y_a(x).
\end{align*}
Here the evaluation result of 
an atom $a$ is memorized in all subsequent evaluations
of $a$ and all subsequent $a$'s can be contracted. 

Typical axiom (right-oriented version):
\[
x\lef y\rig(z\lef u\rig(v\lef y\rig w))= 
x\lef y\rig(z\lef u\rig w).
\]
Typical counter-example: $a\lef b\rig\fa\neq b\lef a\rig\fa$
(thus $b\leftand a\neq a\leftand b$).
\item The variety of \RVA s with \emph{static} 
valuations: 
all \RVA s that satisfy for all $a,b\in A$,
\[y_a(\dd b(x))=y_a(x)\]

This is the 
setting of conventional propositional
logic. 

Typical identities:
$a=b\andthen a$ and
$a\lef b\rig\fa = b\lef a\rig\fa$ 
(thus $b\leftand a= a\leftand b$).
\end{enumerate}

\begin{definition}
\label{def:eqval}
Let $K$ be a variety of valuation algebras over $A$. 
Then propositional statements $P$ and $Q$ are 
\textbf{$K$-equivalent}, notation
\[P\equiv_K Q,\]
if $P/H=Q/H$ for all $\BA\in K$ and $H\in \BA$.
Let $=_K$ be the largest 
congruence contained in $\equiv_K$.
Propositional statements $P$ and $Q$ are
\textbf{$K$-congruent} if
\[P=_K Q.\]
\end{definition}

So, by the varieties defined thus far we distinguish six
types of $K$-equivalence and $K$-congruence: 
free, repetition-proof, 
contractive, weakly memorizing, memorizing and static. 
We use the following
abbreviations for these: 
\[K=\fr,{\rp},{\con},{\wm},{\mem},{\st},\]
respectively. A convenient auxiliary
notation in comparing these equivalences and congruences
concerns the valuation of strings: 
given a \RVA, say $\BA$, a  
valuation $H\in\BA$ can be associated with a function
$H_f: A^+\rightarrow \BB$
by defining
\[
  H_f(a)=y_a(H)\quad\text{and}\quad
  H_f(a\sigma)=(\dd a (H))_f(\sigma).
\]

\begin{proposition}
The inclusions $\equiv_{\fr}\;\subseteq\;\equiv_{\rp}\;\subseteq\;
\equiv_{\con}\;\subseteq\;\equiv_{\wm}\;\subseteq\;\equiv_{\mem}\;
\subseteq\;\equiv_{\st}$, and
$=_{K}\;\subseteq\;\equiv_{K}$ for 
$K\in\{\fr,\rp,\con,\wm,\mem\}$
are all proper.
\end{proposition}

\begin{proof}
In this proof we assume that all \RVA s we use satisfy
$\tr\ne\fa$. We first consider the differences between
the mentioned equivalences:
\begin{enumerate}
\item
 $a\equiv_{\rp}a\lef a\rig\fa$, but $\equiv_{\fr}$ 
does not hold in this case as is witnessed by a \RVA\ 
with valuation $H$ that satisfies
$H_f(a)=\tr$ and $H_f(aa)=\fa$ (yielding
$a/H=\tr$ and $(a\lef a\rig\fa)/H=\fa$).

\item
$b\lef a\rig\fa\equiv_{\con}b\lef(a\lef a\rig\fa)\rig\fa$,
but $\equiv_{\rp}$ does not hold in 
the \RVA\ with element $H$ with
$H_f(a)=H_f(ab)=\tr$ and $H_f(aab)=\fa$.

\item
$(a\lef b\rig \fa)\lef a\rig\fa\equiv_{\wm}
b\lef a\rig\fa$, but $\equiv_{\con}$ does not hold in 
the \RVA\ with element $H$ with
$H_f(a)=H_f(ab)=\tr$ and $H_f(aba)=\fa$.

\item
$(\tr\lef b\rig (\fa\lef a\rig\tr))\lef a\rig\fa\equiv_{\mem}
b\lef a\rig\fa$, 
but $\equiv_{\wm}$ does not hold in 
the \RVA\ with element $H$ with
$H_f(a)=\tr$ and $H_f(ab)=H_f(aba)=\fa$.

\item
$a\equiv_{\st}a\lef b\rig a$
(distinguish all possible cases), but 
$\equiv_{\mem}$ does not hold as is witnessed
by the \RVA\ with element $H$ with
$H_f(a)=H_f(b)=\tr$ and $H_f(ba)=\fa$ (yielding 
$a/H=\tr$ and $(a\lef b\rig a)/H=\fa$). 
\end{enumerate}
Finally, for $K\in\{\fr,\rp,\con,\wm,\mem\}$, 
\(\tr  \equiv_{K}\tr \lef a \rig \tr,\)
but
$b\lef \tr\rig \tr\not\equiv_{K}
b\lef (\tr \lef a \rig \tr)\rig \tr$ 
as is
witnessed by the \RVA\ with element $H$ with
$H_f(a)=H_f(b)=\tr$ and $H_f(ab)=\fa$.
\end{proof}

The following proposition stems from~\cite{Chris} and 
can be used to deal with 
the difference between $K$-congruence and $K$-equivalence.

\begin{proposition}
\label{prop:3}
If $P\equiv_K Q$ and for all $\BA\in K$ and $H\in\BA$,
$\dd P(H)=\dd Q(H)$,
then $P=_K Q$.
\end{proposition}

\begin{proof} Assume $P\equiv_K Q$ 
and the further requirement in the proposition is satisfied.
Since $=_K$ is 
defined as the largest congruence contained in 
$\equiv_K$, $P =_K Q$ if for all $S$ and $R$ 
the following three cases are true: 
\begin{align*}
P\lef S\rig R 
&\equiv_K Q\lef S\rig R,&
\\ 
S\lef P\rig R 
&\equiv_K S\lef Q\rig R,&
\\ 
S\lef R\rig P 
&\equiv_K S\lef R\rig Q.&
\end{align*}
The first and last case follow immediately. For the middle
case, derive
\begin{align*}
(S\lef P\rig R )/H&=
\begin{cases}
S/\dd P(H)&\text{if }P/H=\tr,\\
R/\dd P(H)&\text{if }P/H=\fa,
\end{cases}
\\
&=\begin{cases}
S/\dd Q(H)&\text{if }Q/H=\tr,\\
R/\dd Q(H)&\text{if }Q/H=\fa,
\end{cases}
\\
&=(S\lef Q\rig R )/H.
\end{align*}
\end{proof}

\section{Completeness for $=_{\rp}$ and
$=_{\con}$}
\label{sec:com}
In this section we give complete axiomatizations of 
repetition-proof valuation congruence and of
contractive valuation congruence.
We start with a basic result on free
valuation congruence
of basic forms.

\begin{lemma}
\label{lem:above}
For all basic forms $P$ and $Q$, $P=_{\fr}Q$ implies
$P\equiv Q$.
\end{lemma}

\begin{proof}
We prove the lemma by structural 
induction on $P$ and $Q$.

If $P\equiv\tr$ then if $Q\equiv\tr$ there is nothing to prove.
If $Q\equiv\fa$ then $P\neq_{\fr} Q$ 
(consider a \RVA\ in which $\tr\neq\fa$). If 
$Q\equiv Q_1\lef a\rig Q_2$ then consider a \RVA\ with 
$\tr\neq\fa$ and with an
element $H$ that satisfies 
$y_b(H)=\tr$ and
$y_b(\dd R(\dd a(H)))=\fa$ for all $R$. Assume
$\tr=_{\fr}Q$, then also $(\tr\andthen b)=_{\fr}
((Q_1\lef a\rig Q_2)\andthen b)$ and hence 
$(\tr\andthen b)/H=\tr$ while $(Q\andthen b)/H=
y_b(\dd{Q_1}(\dd a(H)))=\fa$ which is a contradiction, hence
$\tr\neq_{\fr} Q$.

If $P\equiv\fa$ a similar argument applies.

If $P\equiv P_1\lef a\rig P_2$ then the cases $Q\equiv
\tr$ and $Q\equiv\fa$ 
can de dealt with as above. 
If $Q\equiv Q_1\lef b \rig Q_2$ we find for $a=b$
by induction the desired result, and otherwise  
consider a \RVA\ with $\tr\neq\fa$
and with an element $H$ that satisfies 
$y_a(H)=y_b(H)=\tr$,
$y_a(\dd R(\dd a(H)))=\fa$ and
$y_a(\dd R(\dd b(H)))=\tr$ for all $R$. Assume
$P=_{\fr}Q$, then also $(P\andthen a)/H=
(Q\andthen a)/H$ and hence $\fa=
y_a(\dd{P_1}(\dd a(H)))=
y_a(\dd{Q_1}(\dd b(H)))=\tr$, which is a contradiction, hence
$P\neq_{\fr} Q$.
\end{proof}

As a corollary we find a proof of Proposition~\ref{propo},
i.e., for basic forms,
provable equality in \CP\ and syntactic equality coincide:
\begin{proof}[Proof of Proposition~\ref{propo}.]
By soundness (Theorem~\ref{thm:sound})
it is sufficient to prove that for all
basic forms $P$ and $Q$, $P=_{\fr}Q$ implies
$P\equiv Q$, and this is proved in Lemma~\ref{lem:above}.
\end{proof}

It now easily follows that \CP\ axiomatizes
free valuation congruence:
\begin{theorem}[Completeness]
\label{thm:complete}
If $P=_{\fr}Q$ for propositional statements $P$ and $Q$,
then $\CP\vdash P=Q$.
\end{theorem}

\begin{proof}
Assume $P=_{\fr}Q$.
By Lemma~\ref{proplem} there are basic forms $P'$ and
$Q'$ with $\CP\vdash P=P'$ and $\CP\vdash Q=Q'$.
By soundness, $P'=_{\fr}Q'$ and by
Proposition~\ref{propo} $P'\equiv Q'$. Hence,
$\CP\vdash P=Q$.
\end{proof}

We proceed by discussing completeness results for the
other valuation varieties introduced in the previous section.

\begin{theorem}
Repetition-proof valuation congruence $=_{\rp}$ is 
axiomatized by the axioms in \CP\ (see Table~\ref{CP})
and these axiom schemes 
($a\in A$):
\begin{align*}
(\axname{CPrp1})\qquad
(x\lef a\rig y)\lef a\rig z&=(x\lef a\rig x)\lef a\rig z,\\
(\axname{CPrp2})\qquad
x\lef a\rig (y\lef a\rig z)&=x\lef a\rig (z\lef a \rig z).
\end{align*}
\end{theorem}

\begin{proof}
Let \BA\ be a \RVA\ in the variety $\rp$ of repretition-proof
valuations, thus
\[y_a(x)=y_a(\dd a(x)).\]
Concerning soundness, we only
check axiom scheme \axname{CPrp2} (the proof for
\axname{CPrp1} is very similar): let $H\in\BA$, then
\begin{align*}
(P\lef a\rig (Q\lef a\rig R))/H
&=\begin{cases}
P/\dd a(H)&\text{if $y_a(H)=\tr$},\\
(Q\lef a\rig R)/\dd a(H)&\text{if $y_a(H)=\fa$},
\end{cases}\\
&=\begin{cases}
P/\dd a(H)&\text{if $y_a(H)=\tr$},\\
R/\dd a(\dd a H))&\text{if $y_a(H)=\fa=y_a(\dd a (H)$},
\end{cases}
\\
&=(P\lef a\rig (R\lef a\rig R))/H,
\end{align*}
and
\begin{align*}
\dd {(P\lef a\rig (Q\lef a\rig R))}(H)
&=\begin{cases}
\dd P(\dd a(H))&\text{if $y_a(H)=\tr$},\\
\dd{(Q\lef a\rig R)}(\dd a(H))&\text{if $y_a(H)=\fa$},
\end{cases}\\
&=\begin{cases}
\dd P(\dd a(H))&\text{if $y_a(H)=\tr$},\\
\dd R(\dd a(\dd a H)))&\text{if $y_a(H)=\fa=y_a(\dd a (H)$},
\end{cases}
\\
&=\dd{(P\lef a\rig (R\lef a\rig R))}(H).
\end{align*}

In order to prove completeness we use a variant 
of basic forms, which we call
\rp-basic forms, that `minimizes' on
repetition-proof valuation congruence: 
\begin{itemize}
\item $\tr$ and $\fa$ are \rp-basic
forms, and
\item $P_1\lef a\rig P_2$ is an \rp-basic form if
$P_1$ and $P_2$ are \rp-basic forms, and if $P_i$ is not 
equal to $\tr$ or $\fa$, then either
the central
condition in $P_i$ is different from $a$, or 
$P_i$ is of the form $a\andthen P'$ with $P'$ an 
\rp-basic form. 
\end{itemize}
Each propositional statement can in $\CP_{\rp}$
be proved equal to an \rp-basic form (this follows
by structural induction). 
For $P$ and $Q$ \rp-basic forms, 
$P =_{\rp} Q$ implies $P\equiv Q$.
This follows in the same way as in the proof of 
Lemma~\ref{lem:above}. 

Assume $P=_{\rp}Q$, so there exist
\rp-basic forms $P'$ and
$Q'$ with $\CP_{\rp}\vdash P=P'$ and $\CP_{\rp}\vdash Q=Q'$.
By soundness, $P'=_{\fr}Q'$ and as argued above,
$P'\equiv Q'$. Hence, $\CP_{\rp}\vdash P=Q$.
\end{proof}

\begin{theorem}
Contractive valuation congruence $=_{\con}$
is axiomatized by the axioms in \CP\
(see Table~\ref{CP}) and these axiom schemes 
($a\in A$):
\begin{align*}
(\axname{CPcr1})\qquad
(x\lef a\rig y)\lef a\rig z&=x\lef a\rig z,\\
(\axname{CPcr2})\qquad
x\lef a\rig (y\lef a\rig z)&=x\lef a\rig z.
\end{align*}
\end{theorem}
These schemes contract for each $a\in A$
respectively the \tr-case and the \fa-case,
and immediately imply \axname{CPrp1}
and \axname{CPrp2}.

\begin{proof}
Let \BA\ be a \RVA\ in the variety $\con$ of contractive 
valuations, i.e., for all $a\in A$,
\[\dd a(\dd a(x))
=\dd a(x)
\qquad\text{and}\qquad y_a(x)=y_a(\dd a(x)).
\]
Concerning soundness we only check axiom scheme 
\axname{CPcr1}: let $H\in \BA$, then
\begin{align*}
((P\lef a\rig Q)\lef a\rig R)/H
&=\begin{cases}
(P\lef a\rig Q)/\dd a(H)&\text{if $y_a(H)=\tr$},\\
R/\dd a(H)&\text{if $y_a(H)=\fa$},
\end{cases}\\
&=\begin{cases}
P/\dd a(H)&\text{if $y_a(H)=\tr=y_a(\dd a (H))$},\\
R/\dd a(H)&\text{if $y_a(H)=\fa$},\\
\end{cases}\\
&=(P\lef a\rig R)/H,
\end{align*}
and
\begin{align*}
\dd{((P\lef a\rig Q)\lef a\rig R)}(H)
&=\begin{cases}
\dd{(P\lef a\rig Q)}(\dd a(H))&\text{if $y_a(H)=\tr$},\\
\dd R(\dd a(H))&\text{if $y_a(H)=\fa$},
\end{cases}\\
&=\begin{cases}
\dd P(\dd a(\dd a(H)))&\text{if $y_a(H)=\tr=y_a(\dd a (\dd a(H)))$},\\
\dd R(\dd a(H))&\text{if $y_a(H)=\fa$},\\
\end{cases}\\
&=\begin{cases}
\dd P(\dd a(H))&\text{if $y_a(H)=\tr=y_a(\dd a (H))$},\\
\dd R(\dd a(H))&\text{if $y_a(H)=\fa$},\\
\end{cases}\\
&=\dd {(P\lef a\rig R)}(H).
\end{align*}

In order to prove completeness we again use a variant 
of basic forms, which we call
\con-basic forms, that `minimizes' on
contractive valuation congruence: 
\begin{itemize}
\item $\tr$ and $\fa$ are \con-basic
forms, and
\item $P_1\lef a\rig P_2$ is a \con-basic form if
$P_1$ and $P_2$ are \con-basic forms, and if $P_i$ is not 
equal to $\tr$ or $\fa$, the central
condition in $P_i$ is different from $a$. 
\end{itemize}
Each propositional statement can in $\CP_{\con}$
be proved equal to a \con-basic form (this follows
by structural induction). 
For $P$ and $Q$ \con-basic forms, 
$P =_{\con} Q$ implies $P\equiv Q$.
This follows in the same way as in the proof of 
Lemma~\ref{lem:above}. 

Assume $P=_{\con}Q$, so there exist
\con-basic forms $P'$ and
$Q'$ with $\CP_{\con}\vdash P=P'$ and $\CP_{\con}\vdash Q=Q'$.
By soundness, $P'=_{\con}Q'$ and as argued above,
$P'\equiv Q'$. Hence, $\CP_{\con}\vdash P=Q$.
\end{proof}

\section{Completeness for $=_{\wm}$}
\label{sec:wm}
In this section we give a complete axiomatization of 
weakly memorizing valuation congruence.

\begin{theorem}
\label{thm:wm}
Weakly memorizing valuation congruence $=_{\wm}$
is axiomatized by the axioms in $\CP_{\con}$ 
and these axiom schemes ($a,b\in A$):
\begin{align*}
(\axname{CPwm1})\qquad
((x\lef a\rig y)\lef b\rig z)\lef a\rig v
&=(x\lef b\rig z)\lef a\rig v,\\
(\axname{CPwm2})\qquad
x\lef a\rig(y\lef b\rig(z\lef a\rig v))
&=x\lef a\rig(y\lef b\rig v)
\end{align*}
We write $\CP_{\wm}$ for this set of axioms.
\end{theorem}
Before giving a proof we
discuss some characteristics of $\CP_{\wm}$. 
We define a special type of basic forms. 

Let $P$ be a
basic form. Define $pos(P)$ as
the set of
atoms that occur at left-hand (positive) positions in 
 $P$:
$pos(\tr)=pos(\fa)=\emptyset$ and $pos(P\lef a\rig Q)=\{a\}\cup pos(P)$, 
and define $neg(P)$ as the set of atoms 
that occur at right-hand (negative) positions in $P$:
$neg(\tr)=neg(\fa)=\emptyset$ and $neg(P\lef a\rig Q)=\{a\}\cup neg(Q)$. 

Now \wm-basic forms are defined as follows:
\begin{itemize}
\item $\tr$ and $\fa$ are \wm-basic
forms, and
\item $P\lef a\rig Q$ is a \wm-basic form if
$P$ and $Q$ are \wm-basic forms and $a\not \in pos(P)\cup
neg(Q)$.
\end{itemize}
The idea is that in a \wm-basic form, as long as the evaluation 
of consecutive atoms keeps yielding the same reply, no atom 
is evaluated twice. Clearly, each \wm-basic form also is a 
\con-basic form, but not vice versa, e.g., 
$\tr\lef a\rig(\tr\lef b\rig(\tr\lef a\rig\fa))$ is not a 
\wm-basic form because $a\in neg(\tr\lef b\rig(\tr\lef a\rig\fa))$. 
A more intricate example is one in which $a$ and
$b$ ``alternate":
\[(\tr\lef b\rig[(\fa\lef b\rig(\tr\lef a\rig\fa))\lef a\rig\tr])\lef a\rig\fa\]
is a \wm-basic form because 
$pos(\tr\lef b\rig[(\fa\lef b\rig(\tr\lef a\rig\fa))\lef a\rig\tr])
=\{b\}\not\ni a$ and
$\tr\lef b\rig[(\fa\lef b\rig(\tr\lef a\rig\fa))\lef a\rig\tr]$
is a \wm-basic form, where the latter statement follows
because 
$neg((\fa\lef b\rig(\tr\lef a\rig\fa))\lef a\rig\tr)=\{a\}\not\ni b$ 
and because
$\fa\lef b\rig(\tr\lef a\rig\fa)$ is
clearly a \wm-basic form.

\begin{proposition}
\label{prop:wm}
For each propositional statement $P$ there is a \wm-basic form
$P'$ with $\CP_{\wm}\vdash P=P'$.
\end{proposition}

\begin{proof}
By Proposition~\ref{propo} we may assume that $P$ is a 
basic form and we proceed by structural induction on $P$. 
If $P=\tr$ or $P=\fa$ we are done. Otherwise,
$P=P_1\lef a\rig P_2$ and we may assume that $P_i$ are
\wm-basic forms (if not, they can proved equal to \wm-basic forms). 
We first consider the positive side of $P$. If $a\not\in
pos(P_1)$ we are done, otherwise we saturate
$P_1$ by replacing each atom $b\ne a$ that occurs in a positive
position with $(a\lef b\rig\fa)$ 
using axiom \axname{CPwm1}.
After this way we can retract 
each $a$ that is in $pos(P_1)$
(also using \axname{CPcr1}) and end up with $P_1'$ that does 
not contain $a$ on positive positions. For example,
\begin{align*}
&(((\tr\lef a\rig R)\lef b\rig S)\lef c\rig V)\lef a \rig P_2\\
&~=(((\tr\lef a\rig R)\lef (a\lef b\rig\fa)\rig S)\lef (a\lef c\rig\fa)\rig V)\lef a \rig P_2\\
&~=(((((\tr\lef a\rig R)\lef a\rig S)\lef b\rig S)\lef a\rig V)\lef c\rig V)\lef a\rig P_2\\
&~=(((\tr\lef b\rig S)\lef a\rig V)\lef c\rig V)\lef a\rig P_2\\
&~=((\tr\lef b\rig S)\lef c\rig V)\lef a \rig P_2.
\end{align*}

Following the same procedure
for the negative side of $P$ (saturation with 
($\tr\lef b\rig a)$ 
for all $b\ne a$ etc.) yields a \wm-basic form
$P_1'\lef a\rig P_2'$ with 
$\CP_{\wm}\vdash P=P_1'\lef a\rig P_2'$.
\end{proof}

We state without proof:
\begin{proposition}
\label{prop:wm2}
For all \wm-basic forms $P$ and $Q$, $P=_{\wm} Q$
implies $P\equiv Q$.
\end{proposition}

\begin{proof}[Proof of Theorem~\ref{thm:wm}.]
Let \BA\ be a \RVA\ in the variety $\wm$ of weakly 
memorizing valuations, thus for all $a,b\in A$,
\[y_b(\dd a (x))=y_a(x)\rightarrow[\dd a(\dd b(\dd a(x)))
=\dd b(\dd a(x))
~\wedge~
y_a(\dd b(\dd a(x)))=y_a(x)],\]
and all equations for contractive valuations hold in \BA.
The soundness of axiom \axname{CPwm1} and
\axname{CPwm2} follows immediately. We prove the latter one:
let $H\in\BA$, then
\begin{align*}
(P\lef a\rig(&Q\lef b\rig(Z\lef a\rig V)))/H\\
&=\begin{cases}
P/\dd a(H)&\text{if $y_a(H)=\tr$},\\
Q/\dd b(\dd a(H))&\text{if $y_a(H)=\fa$ and $y_b(\dd a(H))=\tr$},\\
(Z\lef a\rig V)/\dd b(\dd a(H))&\text{if $y_a(H)=\fa$ and $y_b(\dd a(H))=\fa$},
\end{cases}\\
&=\begin{cases}
P/\dd a(H)&\text{if $y_a(H)=\tr$},\\
Q/\dd b(\dd a(H))&\text{if $y_a(H)=\fa$ and $y_b(\dd a(H))=\tr$},\\
V/\dd b(\dd a(H))&\text{if $y_a(H)=\fa$ and $y_b(\dd a(H))=\fa$},
\end{cases}\\
&=(P\lef a\rig (Q\lef b\rig V))/H,
\end{align*}
and \quad$\dd{(P\lef a\rig(Q\lef b\rig(Z\lef a\rig V)))}(H)$
\begin{align*}
\hspace{15mm}&=\begin{cases}
\dd P(\dd a(H))&\text{if $y_a(H)=\tr$},\\
\dd Q(\dd b(\dd a(H)))&\text{if $y_a(H)=\fa$ and $y_b(\dd a(H))=\tr$},\\
\dd {(Z\lef a\rig V)}(\dd b(\dd a(H)))&\text{if $y_a(H)=\fa$ and $y_b(\dd a(H))=\fa$},
\end{cases}\\
&=\begin{cases}
\dd P(\dd a(H))&\text{if $y_a(H)=\tr$},\\
\dd Q(\dd b(\dd a(H)))&\text{if $y_a(H)=\fa$ and $y_b(\dd a(H))=\tr$},\\
\dd V(\dd b(\dd a(H)))&\text{if $y_a(H)=\fa$ and $y_b(\dd a(H))=\fa$},
\end{cases}\\
&=\dd {(P\lef a\rig (Q\lef b\rig V))}(H).
\end{align*}

In order to prove completeness assume $P=_{\wm} Q$. 
By Proposition~\ref{prop:wm} there are
\wm-basic forms $P'$ and $Q'$ with $\CP_{\wm}\vdash P=P',~Q=Q'$. 
By soundness $P'=_{\wm}Q'$ and by Proposition~\ref{prop:wm2},
$P'\equiv Q'$, 
and thus $\CP_{\wm}\vdash P=Q$.
\end{proof}

\section{Completeness for $=_{\mem}$}
\label{sec:mem}

In this section we provide a complete axiomatization of 
memorizing valuation congruence.

\begin{theorem}
\label{thm:mem}
Memorizing valuation congruence $=_{\mem}$
is axiomatized by the axioms in $\CP$ (see Table~\ref{CP}) 
and this axiom:
\begin{align*}
(\axname{CPmem}) \qquad
x\lef y\rig(z\lef u\rig(v\lef y\rig w))&= x\lef y\rig(z\lef u\rig w).
\end{align*}
We write $\CP_{\mem}$ for this set of axioms.
\end{theorem}
Before proving this theorem, we
discuss some characteristics of $\CP_{\mem}$. 
Axiom \axname{CPmem}
defines how the central condition $y$ may recur in an 
expression. This axiom yields in combination 
with \CP\ some interesting consequences. 
First, \axname{CPmem} has three symmetric variants, which
all follow easily with 
$x\lef y\rig z=z\lef(\fa\lef y\rig\tr)\rig x$ ($=z\lef\neg y\rig x$):
\begin{align}
\label{mini}
x\lef y\rig ((z\lef y\rig u)\lef v\rig w)&=x\lef y\rig (u\lef v\rig w),\\
\label{mini1}
(x\lef y\rig (z\lef u\rig v)) \lef u\rig w&=(x\lef y\rig z)\lef u\rig w,\\
\label{mini2}
((x\lef y\rig z)\lef u\rig v) \lef y\rig w&=(x\lef u\rig v)\lef y\rig w.
\end{align} 
The axioms of $\CP_{\mem}$ imply various laws for contraction:
\begin{align}
\label{mini3}
x\lef y\rig(v\lef y\rig w)&=x\lef y\rig w
&&\text{(take $u=\fa$ in \axname{CPmem})},
\\
\label{mini4}
x\lef y\rig(\tr\lef u\rig y)&=x\lef y\rig u
&&\text{(take $z=v=\tr$ and $w=\fa$ in \axname{CPmem})},
\\
\label{mini5}
(x\lef y\rig z)\lef y\rig u&= x \lef y\rig u,
\\
\label{qquatre}
(x\lef y\rig \fa)\lef x\rig z&= y \lef x\rig z,
\end{align} 
and thus (take $v=\tr$ and $w=\fa$ in \eqref{mini3},
respectively $x=\tr$ and $z=\fa$ in \eqref{mini5}),
\[
x\lef y\rig y=x\lef y\rig \fa\quad\text{and}\quad
y\lef y\rig u=\tr\lef y\rig u.
\]
The latter two equations
immediately imply the following very simple contraction laws: 
\[
x\lef x\rig x=x\lef x\rig\fa=\tr\lef x\rig x=\tr\lef x\rig 
\fa=x.\]

\bigskip

Let $A'$ be a finite subset of $A$.
We employ a special type of basic forms
based on $A'$: {\mem-basic forms over $A'$}
are defined by
\begin{align*}
\tr,\fa,\\
P\lef a \rig Q&\quad \text{if $a\in A'$ and $P$ and $Q$ are 
\mem-basic forms over $A'\setminus\{a\}$.}
\end{align*}
E.g., for $A'=\{a\}$ the set of all \mem-basic forms is 
$\{bv,~bv\lef a\rig bv'\mid 
bv,bv'\in\{\tr,\fa\}\}$, and for $A'=\{a,b\}$ it is
\begin{align*}
\{bv,t_1\lef a\rig t_2,t_3\lef b\rig t_4\mid 
bv\in\{\tr,\fa\},~&t_1,t_2 \text{ \mem-basic forms over $\{b\}$,}\\
&\text{$t_3,t_4$  \mem-basic forms over $\{a\}$}\}.
\end{align*}

\begin{proposition}
\label{prop:mem}
For each propositional statement $P$ there is a \mem-basic form
$P'$ with $\CP_{\mem}\vdash P=P'$.
\end{proposition}

\begin{proof}
By Proposition~\ref{propo} we may assume that $P$ is a 
basic form and we proceed by structural induction on $P$. 
If $P=\tr$ or $P=\fa$ we are done. Otherwise,
$P=P_1\lef a\rig P_2$. 

We first show that
\[\CP_{\mem}\vdash P_1\lef a\rig P_2=P_1[\tr/a]\lef a\rig P_2\]
by induction on $P_1$: if $P_1$ equals $\tr$ of $\fa$ this is clear.
If $P_1=Q\lef a\rig R$ then $\CP\vdash P_1[\tr/a]=Q[\tr/a]$ and we
derive
\begin{eqnarray*}
P_1\lef a\rig P_2&=&(Q\lef a\rig R)\lef a\rig P_2\\
&
\stackrel{IH}=&(Q[\tr/a]\lef a\rig R)
\lef a\rig P_2\\
&
\stackrel{\eqref{mini5}}=&Q[\tr/a]\lef a\rig P_2\\
&=&P_1[\tr/a]\lef a\rig P_2,
\end{eqnarray*}
and if $P_1=Q\lef b\rig R$ with $b\neq a$ then $\CP\vdash
P_1[\tr/a]=Q[\tr/a]\lef b\rig R[\tr/a]$ and we derive
\begin{eqnarray*}
P_1\lef a\rig P_2&=&(Q\lef b\rig R)\lef a\rig P_2\\
&\stackrel{\eqref{mini1}\eqref{mini2}}=&((Q\lef a\rig \tr)
\lef b\rig(R\lef a\rig\tr))
\lef a\rig P_2\\
&
\stackrel{IH}=&((Q[\tr/a]\lef a\rig \tr)
\lef b\rig(R[\tr/a]\lef a\rig\tr))
\lef a\rig P_2\\
&
\stackrel{\eqref{mini1}\eqref{mini2}}=&(Q[\tr/a]\lef b
\rig R[\tr/a])
\lef a\rig P_2\\
&=&P_1[\tr/a]\lef a\rig P_2.
\end{eqnarray*}
In a similar way, but now using \eqref{mini3}, axiom \axname{CPmem} and \eqref{mini}
instead, we find 
$\CP_{\mem}\vdash P_1\lef a\rig P_2=P_1\lef a\rig P_2[\fa/a]$, and thus
\[\CP_{\mem}\vdash P_1\lef a\rig P_2=P_1[\tr/a]\lef a\rig P_2[\fa/a].\]
Finally, with axioms \axname{CP1} and \axname{CP2} we find
basic forms $Q_i$ in which $a$ does not occur with 
$\CP_{\mem}\vdash Q_1=P_1[\tr/a]$ and $\CP_{\mem}\vdash 
Q_2=P_2[\fa/a]$. 
\end{proof}

We state without proof:
\begin{proposition}
\label{prop:mem2}
For all \mem-basic forms $P$ and $Q$, $P=_{\mem} Q$
implies $P\equiv Q$.
\end{proposition}

\begin{proof}[Proof of Theorem~\ref{thm:mem}.]
Let \BA\ be a \RVA\ in the variety $\mem$ of memorizing 
valuations, i.e., for all $a,b\in A$,
\[\dd a(\dd b(\dd a(x)))
=\dd b(\dd a(x))
~\wedge~
y_a(\dd b(\dd a(x)))=y_a(x),\]
and all equations for contractive 
valuations also hold in \BA.
In order to prove soundness we use
the following generalization of these equations, which can 
easily be proved by structural induction on $Q$:
for all propositional statements $P,Q,R,S$ and valuations $H\in\BA$,
\begin{equation}
\label{eq:mem}
\dd Q(\dd P(\dd Q(H)))=\dd P(\dd Q(H))\quad\text{and}\quad
Q/\dd P(\dd Q(H))=Q/H.\end{equation}
If $Q=a$, then apply induction on $P$: the cases $P=\tr$ and 
$P=\fa$ are trivial, and if $P=b$ then \eqref{eq:mem}
follows by definition. If
$P=P_1\lef P_2\rig P_3$ then if $P_2/\dd a(H)=\tr$,
\begin{align*}
\dd a(\dd{(P_1\lef P_2\rig P_3)}(\dd a(H)))
&=\dd a(\dd{P_1}(\dd{P_2}(\dd a(H))))\\
&=\dd a(\dd{P_1}(\dd a(\dd{P_2}(\dd a(H)))))\\
&=\dd{P_1}(\dd a(\dd{P_2}(\dd a(H))))\\
&=\dd{P_1}(\dd{P_2}(\dd a(H)))\\
&=\dd{(P_1\lef P_2\rig P_3)}(\dd a(H)), 
\end{align*}
and
\begin{align*}
a/\dd{(P_1\lef P_2\rig P_3)}(\dd a(H))
&=a/\dd{P_1}(\dd{P_2}(\dd a(H)))\\
&=a/\dd{P_1}(\dd a(\dd{P_2}(\dd a(H))))\\
&=a/\dd{P_2}(\dd a(H))\\
&=a/H,
\end{align*}
and if $P_2/\dd a(H)=\fa$ a similar argument applies.

If $Q=Q_1\lef Q_2\rig Q_3$ then first assume $Q_2/H=\tr$,
so 
\begin{align*}
\dd Q(\dd P(\dd Q(H)))&=
\dd{Q_1}(\dd{(P\andthen Q_2)}(\dd{Q_1}(\dd{Q_2}(H))))\\
&=\dd{(P\andthen Q_2)}(\dd{Q_1}(\dd{Q_2}(H)))\\
&=\dd{Q_2}(\dd{(Q_1\andthen P)}(\dd{Q_2}(H)))\\
&=
\dd{P}(\dd{Q_1}(\dd{Q_2}(H)))=\dd P(\dd Q(H)), 
\end{align*}
and
$(Q_1\lef Q_2\rig Q_3)/\dd P(\dd Q(H))=
Q_1/\dd{P\andthen Q_2}(\dd{Q_1}(\dd{Q_2}(H))
=Q_1/\dd{Q_2}(H)=
Q/H$.
Finally, 
if $Q_2/H=\fa$ a similar argument applies.
\bigskip

We find by \eqref{eq:mem}
for $Q/H=\fa$ and propositional
statements $V,W,S$ that
\[
(V\lef Q\rig W)/\dd S(\dd Q(H))=W/\dd S(\dd Q(H))
~~\text{and}~~
\dd{(V\lef Q\rig W)}(\dd S(\dd Q(H)))=\dd W(\dd S(\dd Q(H))).
\]
Now the soundness of axiom \axname{CPmem} follows immediately:
\begin{align*}
(P\lef Q\rig(&R\lef S\rig(V\lef Q\rig W)))/H\\
&=\begin{cases}
P/\dd Q(H)&\text{if $Q/H=\tr$},\\
R/\dd S(\dd Q(H))&\text{if $Q/H=\fa$ and $S/\dd Q(H)=\tr$},\\
(V\lef Q\rig W)/\dd S(\dd Q(H))&\text{if $Q/H=\fa$ and $S/\dd Q(H)=\fa$},
\end{cases}\\
&=\begin{cases}
P/\dd Q(H)&\text{if $Q/H=\tr$},\\
R/\dd S(\dd Q(H))&\text{if $Q/H=\fa$ and $S/\dd Q(H)=\tr$},\\
W/\dd S(\dd Q(H))&\text{if $Q/H=\fa$ and $S/\dd Q(H)=\fa$},
\end{cases}\\
&=(P\lef Q\rig (R\lef S\rig W))/H,
\end{align*}
and \quad$\dd{(P\lef Q\rig(R\lef S\rig(V\lef Q\rig W)))}(H)$
\begin{align*}
\hspace{15mm}&=\begin{cases}
\dd P(\dd Q(H))&\text{if $Q/H=\tr$},\\
\dd R(\dd S(\dd Q(H)))&\text{if $Q/H=\fa$ and $S/\dd Q(H)=\tr$},\\
\dd {(V\lef Q\rig W)}(\dd S(\dd Q(H)))&\text{if $Q/H=\fa$ and $S/\dd Q(H)=\fa$},
\end{cases}\\
&=\begin{cases}
\dd P(\dd Q(H))&\text{if $Q/H=\tr$},\\
\dd R(\dd S(\dd Q(H)))&\text{if $Q/H=\fa$ and $S/\dd Q(H)=\tr$},\\
\dd W(\dd S(\dd Q(H)))&\text{if $Q/H=\fa$ and $S/\dd Q(H)=\fa$},
\end{cases}\\
&=\dd {(P\lef Q\rig (R\lef S\rig W))}(H).
\end{align*}

In order to prove completeness assume $P=_{\mem} Q$. 
By Proposition~\ref{prop:mem} there are
\mem-basic forms $P'$ and $Q'$ with $\CP_{\mem}\vdash 
P=P',~Q=Q'$. By soundness, $P'=_{\mem} Q'$, and by
Proposition~\ref{prop:mem2}, $P'\equiv Q'$, 
and thus $\CP_{\mem}\vdash P=Q$.
\end{proof}

\section{Completeness for $=_{\st}$}
\label{sec:stat}

In this section we provide a complete axiomatization of 
static valuation congruence.

\begin{theorem}[Hoare~\cite{Hoa85}]
\label{Hoare}
Static valuation congruence $=_{\st}$
is axiomatized by the axioms in $\CP$ (see Table~\ref{CP}) 
and these axioms:
\begin{align*}
(\axname{CPstat}) \qquad
(x\lef y\rig z)\lef u\rig v&=
(x\lef u\rig v)\lef y\rig(z\lef u\rig v),\\
(\axname{CPcontr}) \qquad
(x\lef y\rig z)\lef y\rig u&= x \lef y\rig u.
\end{align*}
We write $\CP_{\st}$ for this set of axioms.
\end{theorem}
Observe that axiom \axname{CPcontr} equals the
derivable identity \eqref{mini5} which holds in $\CP_{\mem}$.
Also note that the symmetric variants of 
the axioms~\axname{CPstat} and \axname{CPcontr}, say
\begin{align*}
(\axname{CPstat}') \qquad
x\lef y\rig (z\lef u\rig v)&=
(x\lef y\rig z)\lef u\rig(x\lef y\rig v),\\
(\axname{CPcontr}') \qquad
x\lef y\rig(z\lef y\rig u) &= x \lef y\rig u,
\end{align*}
easily follow with identity~\eqref{ss}, i.e.,  
$y\lef x\rig z=z\lef (\fa\lef x\rig\tr)\rig y$,
which is even valid in  
free valuation congruence, and that 
$\axname{CPcontr'}=\eqref{mini3}$. Thus, the axiomatization
of static valuation congruence is obtained from \CP\
by adding the axiom~\axname{CPstat}
that prescribes for a nested conditional composition how
the order of the first and a second central condition can be changed,
and a generalization of the axioms \axname{CPcr1} and \axname{CPcr2} 
that prescribes contraction for terms (instead of atoms).
Moreover, in $\CP_{\st}$ it can be derived that
\begin{align*}
x
&=(x\lef y\rig z)\lef \fa\rig x\\
&=(x\lef\fa\rig x)\lef y\rig (z\lef \fa\rig x)\\
&=x\lef y\rig x\\
&=y\andthen x,
\end{align*}
thus any `and-then' prefix can be added to (or left out from) 
a propositional statement while preserving static valuation
congruence, in particular
\(x\lef x\rig x=x\andthen x=x.\)

\begin{proof}[Proof of Theorem~\ref{Hoare}]
Soundness follows from the definition of static valuations:
let $\BA$ be a \RVA\ that satisfies for all $a,b\in A$,
\[y_a(\dd b (x))=y_a(x).\]
These equations imply that for all $P,Q$ and
$H\in\BA$,
\[P/H=P/\dd Q(H).\]
As a consequence, the validity of axioms~\axname{CPstat} 
and \axname{CPcontr} follows
from simple case distinctions.
Furthermore, Hoare showed in~\cite{Hoa85} that 
$\CP_{\st}$ is complete for static valuation congruence.

For an idea of a direct proof, assume $P=_{\st}Q$ and
assume that the atoms 
occurring in $P$ and $Q$ are ordered as
$a_1,...,a_n$. 
Then under static valuation congruence
each propositional statement containing no other atoms than
$a_1,...,a_n$ can be rewritten into the following
special type of basic form: 
consider the full binary tree with at level 
$i$ only occurrences of atom $a_i$ (there are $2^{i-1}$ such
occurrences), 
and at level $n+1$ only leaves 
that are either $\tr$ or $\fa$
(there are $2^n$ such leaves). 
Then each series of leaves represents 
one of the possible propositional statements in which
these atoms may occur, and the axioms in
$\CP_{\mem}$
are sufficient to rewrite both $P$ and $Q$ into
exactly one such basic form. For these basic forms,
static valuation congruence implies syntactic
equivalence.
Hence, completeness for closed
equations follows.
\end{proof}

As an aside, we note that the axioms~\axname{CPcontr} 
and $\axname{CPcontr}'$
immediately imply \axname{CPcr1} and
\axname{CPcr2}, and conversely, that for $y$ ranging
over closed terms, these axioms are derivable from 
$\CP +\axname{CPstat} +\axname{CPcr1}+\axname{CPcr2}$
(by induction on basic forms), which proves
completeness for closed equations of this particular group 
of axioms. 

\section{Adding negation and definable connectives}
\label{sec:add}
In this section we formally add negation and various definable
connectives to \CP.
As stated earlier (see identity~\eqref{ss}), negation $\neg x$ 
can be defined as follows:
\begin{equation}
\label{nega}
\neg x = \fa\lef x\rig\tr
\end{equation}

\begin{table}[tbp]
\centering
\hrule
\begin{align}
\label{c1}
	\neg \tr &= \fa\\
	\label{c2}
	\neg \fa&=\tr\\
	\label{c3}
	\neg \neg x &= x\\
	\label{c4}
	\neg (x \lef y \rig z) &= \neg x\lef y \rig \neg z\\
	\label{neg}
	x\lef\neg y\rig z&=z\lef y\rig x
\end{align}
\hrule
\caption{Some immediate consequences of the set of axioms \CP\ and
equation~\eqref{nega}}
\label{CPd}
\end{table}

The derivable identities in Table~\ref{CPd} play a role
in the derivable connectives that we discuss below. 
They can be derived as follows:
\begin{enumerate}
\item[\eqref{c1}]
   follows from 
   $\neg \tr=\fa\lef\tr\rig\tr=\fa$,
\item[\eqref{c2}]
   follows in a similar way,
\item[\eqref{c3}]
	follows from
	 $\neg \neg x=\fa\lef(\fa\lef x\rig\tr)
	\rig\tr=(\fa\lef\fa\rig\tr)\lef x\rig(\fa\lef\tr\rig\tr)
    =\tr\lef x\rig\fa=x$,
\item[\eqref{c4}]
    follows in a similar way,
    \item[\eqref{neg}] follows from $x\lef\neg y\rig z=
    (z\lef\fa\rig x)\lef \neg y\rig
    (z\lef\tr\rig x)=z\lef(\fa\lef\neg y\rig\tr)\rig x=
    z\lef y\rig x$.
\end{enumerate}

A definable (binary) connective already introduced 
is the \emph{and then}
operator $\andthen$ with defining equation 
$x\andthen y=y\lef x\rig y$.
Furthermore, following~\cite{BBR95} we write 
\[\leftand\]
for \emph{left-sequential} conjunction, i.e., a 
conjunction that first evaluates its lefthand argument 
and only after that is found $\tr$ carries on with 
evaluating its second argument (the small circle indicates
which argument is evaluated first). Similar notations are used 
for other sequential connectives. 
We provide defining equations for a number of 
derived connectives in Table~\ref{DC}.

\begin{table}[t]
\centering
\hrule
\begin{align*}
	x \leftand y &= y \lef x \rig \fa
	&x \leftimp y &= y\lef x\rig\tr\\\
	x \rightand y &= x \lef y \rig \fa
	&x \rightimp y &= \tr\lef y\rig \neg x\\
	x \leftor y &= \tr \lef x \rig y
	&x \leftbiimp y &= y \lef x \rig \neg y\\
	x \rightor y &= \tr \lef y \rig x
	&x \rightbiimp y &= x \lef y \rig \neg x
\end{align*}
\hrule
\caption{Defining equations for derived connectives}
\label{DC}
\end{table}

The operators $\leftand$ and left-sequential disjunction
$\leftor$ are associative and 
the dual of each other, and so are their right-sequential
counterparts. For $\leftand$ a proof of this is as follows:
\begin{align*}
(x \leftand y)\leftand z
&=z\lef(y \lef x \rig \fa)\rig\fa\\
&=(z\lef y \rig\fa)\lef x \rig (z\lef \fa\rig \fa)\\
&=(y \leftand z)\lef x \rig \fa\\
&=x\leftand(y \leftand z),
\end{align*}
and (a sequential version of De Morgan's laws)
\begin{align*}
\neg(x \leftand y)
&=\fa\lef(y \lef x \rig \fa)\rig\tr\\
&=(\fa\lef y \rig\tr)\lef x \rig (\fa\lef \fa\rig \tr)\\
&=\neg y \lef x \rig \tr\\
&=\tr\lef\neg x\rig\neg y\\
&=\neg x\leftor \neg y.
\end{align*}
Furthermore, note that $\tr\leftand x=x$ and $x\leftand\tr=x$,
and $\fa\leftor x=x$ and $x\leftor\fa=x$.

Of course, \emph{distributivity}, as in
$(x \leftand y)\leftor z=(x\leftor z)\leftand(y\leftor z)$
is not valid in free
valuation congruence: 
it changes the order of evaluation and 
in the right-hand expression $z$ can be evaluated twice.
It is also obvious that both sequential versions of 
\emph{absorption}, one of which reads
\[x = x \rightand(x\rightor y),\]
are not valid. 
Furthermore, it is not difficult to prove in \CP\ that
$\leftbiimp$ and $\rightbiimp$ (i.e., the two sequential 
versions of bi-implication defined in Table~\ref{DC})
are also associative, and that $\leftimp$ and $\rightimp$
are not associative, but
satisfy the sequential versions of the common definition of 
implication:
\[x\leftimp y=\neg x\leftor y\quad\text{and}\quad
x\rightimp y=\neg x\rightor y.\]

From now on we extend 
$\Sigma_\CP(A)$ with the ``and then" operator $\andthen$, 
negation and all derived connectives introduced in this 
section, and we adopt their defining equations. 
Of course, it remains the case that each propositional
statement
has a unique basic form (cf. Lemma~\ref{proplem}).

Concerning the example of the propositional statement
sketched in Example~\eqref{pedes} in Section~\ref{sec:motiv}:
\[\textit{look-left-and-check}\leftand 
\textit{look-right-and-check}\leftand 
\textit{look-left-and-check}
\]
indeed precisely
models part of the processing of a
pedestrian planning to
cross a road with two-way traffic driving on the right. 

\bigskip

We end this section with a brief comment on these connectives
in the setting of other valuation congruences. 
\begin{itemize}
\item In memorizing valuation congruence,
which we call
\emph{Memorizing logic}, the sequential connective
$\leftand$ has the following properties:
\begin{enumerate}
\item 
\label{trois}
The associativity of $\leftand$ is valid,
\item 
\label{quatre}
The identity $x\leftand y\leftand x=x\leftand y$ is 
valid (by equation~\eqref{qquatre}),
\item
The connective $\leftand$ is not commutative.
\end{enumerate}

\item
In static valuation congruence, all of $\leftand$,
$\leftor$, $\rightand$ and $\rightor$ are commutative and idempotent.
For example, we derive with axiom \axname{CPstat} that
\[x\leftand y=y\lef x\rig\fa=(\tr\lef y\rig\fa)\lef x\rig\fa=
(\tr\lef x\rig\fa)\lef y\rig(\fa\lef x\rig\fa)=
x\lef y\rig\fa=y\leftand x,
\]
and with axiom \axname{CPcontr} and its symmetric counterpart
$\axname{CPcontr}'$,
\begin{align*}
x\leftor x&=\tr\lef x\rig x=\tr\lef x\rig(\tr\lef x\rig\fa)=
\tr\lef x\rig\fa=x,\\
x\leftand x&=x\lef x\rig\fa=(\tr\lef x\rig x)\lef x\rig\fa=
\tr\lef x\rig\fa=x.
\end{align*}
As a consequence, the sequential notation of these connectives
is not meaningful in this case, and 
distributivity and absorption hold in static valuation congruence.
\end{itemize}

\section{Satisfiability and complexity}
\label{sec:sat}
In this section we briefly consider some complexity issues.
Given a variety $K$ of valuation algebras,
a propositional statement is \emph{satisfiable} with respect to
$K$ if for some non-trivial $\BA\in K$
($\trVal$ and $\faVal$ are different) 
there exists a valuation $H\in\BA$ such that 
\[P/H=\tr.\]
We write 
\[\SAT_K(P)\]
if $P$ is satisfiable. We say that
$P$ is \emph{falsifiable} with respect to
$K$, notation
\[\FAS_K(P),\]
if and only if a valuation $H\in\BA\in K$ exists with 
$P/H=\fa$. This is the case if and
only if $\SAT_K(\fa\lef P\rig\tr)$.

It is a well-known fact that $\SAT_\st$ is an NP-complete problem. 
We now argue that $\SAT_\fr$
is in P. This is the case because for $=_\fr$, both
$\SAT_\fr$ and $\FAS_\fr$ can be simultaneously defined 
in an inductive manner: let $a\in A$ and
write $\neg\SAT_\fr(P)$ to express that $\SAT_\fr(P)$ does not hold, 
and similar for $\FAS_\fr$, then
\begin{align*}
&\SAT_\fr(\tr),&&\neg\FAS_\fr(\tr),
\\
&\neg\SAT_\fr(\fa),&&\FAS_\fr(\fa),
\\
&\SAT_\fr(a),&&\FAS_\fr(a),
\\
\end{align*}
and
\begin{align*}
&\SAT_\fr(P\lef Q\rig R)\quad\text{ if }\quad
\begin{cases}
\SAT_\fr(Q)\text{ and }\SAT_\fr(P),\\
\text{or}\\
\FAS_\fr(Q)\text{ and }\SAT_\fr(R),
\end{cases}
\\\\
&\FAS_\fr(P\lef Q\rig R)\quad\text{ if }\quad
\begin{cases}
\SAT_\fr(Q)\text{ and }\FAS_\fr(P),\\
\text{or}\\
\FAS_\fr(Q)\text{ and }\FAS_\fr(R).
\end{cases}
\end{align*}
Hence, with respect to free valuation congruence both
$\SAT_\fr(P)$ and $\FAS_\fr(P)$ are computable in polynomial time.
In a similar way one can show that both $\SAT_{\rp}$ and $\SAT_{\con}$ are in P.

\bigskip

Of course, many more models of \CP\ exist than those discussed in the previous
sections. For example, call a valuation \emph{positively memorizing}
($\Pmem$) if the reply \tr\ to an atom is preserved after all subsequent replies:
\[x\lef a\rig y = [T/a]x\lef a\rig y\]
for all atomic propositions $a$.
In a similar way one can define negatively memorizing
valuations ($\Nmem$): 
\[x\lef a\rig y = x\lef a\rig[F/a]y.\]
Contractive or weakly memorizing valuations that satisfy $\Pmem$ (\Nmem) give
rise to new models in which more propositional statements are identified. 

\begin{theorem}
For $K\in\{\Pmem, \con{+}\Pmem, \wm{+}\Pmem,\Nmem, \con{+}\Nmem, \wm{+}\Nmem\}$ 
it holds that $\SAT_K$ is NP-complete.
\end{theorem}

\begin{proof}
We only consider the case for \Pmem. Then
\[ \SAT_\st (\phi) = \SAT_\mem(\phi) = 
\SAT_{P\mem}(\phi \leftand...\leftand \phi),\] 
where $\phi$ is repeated $n+1$ times
with $n$ the number of atoms occurring in $\phi$.
Each time a $\phi$ evaluates to \tr\ while it 
would not do so in $=_\mem$, this is due to some atom
that changes the reply.
So, this must be a change from \fa\ to \tr, because
\tr\ remains \tr\ by $=_{P\mem}$.
Per atom this can happen at most once, and if each $\phi$ 
yields \tr, then at least once without
an atom in between flips. But then $\phi$ is also satifiable
in $=_\mem$. 
\end{proof}

For $K\in\{\con{+}P\mem, \wm{+}P\mem, \con{+}N\mem, \wm{+}N\mem\}$,
each closed term can be written with $\tr,\fa$, $\neg$, 
$\leftand$ and $\leftor$
only. For example in $\con{+}\Pmem$: 
\[x\lef a\rig y = (a \leftand x) \leftor(\neg a\leftand y) 
\]
because after a positive reply to $a$ and whatever 
happens in $x$, the next $a$ is again positive, so $y$ is not evaluated,
and after a negative reply to $a$, the subsequent $a$ gets a negative reply 
because of $\con$, so then $y$ is tested.
So here we see models that identify less than $=_\mem$ and in which
each closed term can be written without conditional composition.
At first sight, this cannot be done in a uniform way
(using variables only), and it also yields a combinatoric
explosion because first a rewriting to basic form is needed.
For these models $K$, $\SAT_K$ is known to be NP-complete.

\section{Expressiveness}
\label{sec:exp}
In this section we first show that
the ternary conditional operator cannot be replaced by 
$\neg$ and $\leftand$ and one of $\tr$ and $\fa$
(which together define $\leftor$)
modulo free valuation congruence.
Then we show that this is the case for any collection of
unary and binary operators each of which is
definable in $\Sigma_{\CP}(A)$ with free 
valuation congruence, and in a next theorem
we lift this result to
contractive valuation congruence. Finally we observe
that the conditional operator is definable
with $\neg$ and $\leftand$ and one of $\tr$ and $\fa$
modulo memorizing valuation congruence. It remains
an open question whether this is the case
in weakly memorizing valuation congruence.

An occurrence of an atom $a$ in a propositional statement
over $A$, $\neg$, 
$\leftand$ and $\leftor$ is \emph{redundant} or 
\emph{inaccessible} if it is not used along any of the possible
arbitrary valuations, as in for example $\fa\leftand a$.

The opposite of redundancy is accessibility, which is defined thus
($acc(\phi)\subseteq A$): 
\begin{align*}
acc(a)&=\{a\},\\
acc(\tr)&=\emptyset,\\
acc(\neg x)&=acc(x),\\
acc(x\leftand y)&=
\begin{cases}
acc(x)&\text{if }\neg \SAT_\fr(x),\\
acc(x)\cup acc(y)&\text{if } \SAT\fr(x),
\end{cases}\\
acc(x\leftor y)&=
\begin{cases}
acc(x)&\text{if }\neg \SAT_\fr(\neg x),\\
acc(x)\cup acc(y)&\text{if } \SAT_\fr(\neg x).
\end{cases}
\end{align*}

\begin{proposition}
\label{prop:fr}
Let $a\in A$, then 
the propositional statement
$a\lef a\rig\neg a$ cannot be expressed 
in $\Sigma_{\CP}(A)$ with free valuation
congruence using
$\leftand$, $\leftor$, $\neg$, \tr\ and \fa\ only.
\end{proposition}

\begin{proof}
Let $\psi$ be a minimal expression of $a\lef a\rig\neg a$.

Assume $\psi=\psi_0\leftor\psi_1$. We notice:
\begin{itemize}
\item 
Both $\psi_0$ and $\psi_1$ must contain $a$: if $\psi_0$ 
contains no $a$, it is either $\tr$ and then 
$\psi$ always yields \tr\ which is wrong,
or \fa\ and then $\psi$
can be simplified; if $\psi_1$ contains no $a$ it is either \tr\
and then $\psi$ always yields \tr\ which is wrong, or \fa\ and 
then $\leftor\fa$ can be removed so $\psi$ was not minimal.

\item
$\psi_0$ can yield \fa\ otherwise $\psi$ is not minimal.
It will do so after using exactly one test $a$ 
(yielding \fa\ without a use
of $a$ simply means that $a\not\in acc(\psi_0)$), 
yielding \fa\ after two uses of $a$ implies that evaluation
of $\psi$ has at least three uses of $a$ (which is wrong).
\item 
$\psi_0$ has at most two uses of $a$ if $\psi_0$ yields \tr,
and at most one use of $a$ (so exactly 1) if $\psi_0$ yields \fa.
\end{itemize}
Thus,
$\psi_0=\fa\lef a\rig(\overline a\leftor\tr)$ or 
$\psi_0=\fa\lef \neg a\rig(\overline a\leftor\tr)$, where the 
$\overline a$ in the righthand sides equals either $a$ 
or $\neg a$, and these sides 
take their particular form by minimality 
(other forms are $\tr\leftor \overline a=\tr$, etc.). 
But both are impossible
as both imply that after a first use of $a$ the final value of 
$\psi$ can be independent of the second value returned for $a$
which is not true for $a\lef a\rig\neg a$.

For the case $\psi=\psi_0\leftand\psi_1$ a similar type of 
reasoning applies.
\end{proof}

Below we prove two more general results. We first 
introduce some auxiliary definitions and notations because
we have to be precise about definability by unary
and binary operators.
For $X$ a countable set of variables, we define
\begin{description}
\item[$T_C(X):$] 
the set of terms over 
$X,\tr,\fa,\_\lef\_\rig\_\;$.
\item[$T_{TND}(X):$] 
the set of terms over $X,\tr,\neg,\leftor$.
\item[$T_C^{1,2}(X):$] the smallest set of terms $V$ such that
\begin{itemize}
\item $\tr,\fa\in V$,
\item if $x\in X$ and $t\in T_C(\{x\})$ then $t\in V$,
\item if $x,y\in X$ and $t\in T_C(\{x,y\})$ then $t\in V$,
\item $V$ is closed under substitution.
\end{itemize}
\end{description}
Thus $T_C^{1,2}(X)$ contains the terms that can be made
from unary and binary operators definable in $T_C(X)$.
For $t\in T_C^{1,2}(\{x\})$ 
we sometimes write $t(x)$ instead of $t$, and
if $s\in T_C^{1,2}(X)$ 
we write $t(s)$ for the term
obtained by substituting $s$ for all $x$ in $t$.
Similarly, if $u\in T_C^{1,2}(\{x,y\})$
we may write $u(x,y)$ for $u$ and if $s,s'\in T_C^{1,2}(X)$ 
we write $u(s,s')$ for the term obtained
by substituting $s$ for $x$ and $s'$ for $y$ in $u$.
Finally, we define
$\#_{2p}(t)$ as the number of 2-place terms used in 
the definition of $t$, i.e.,
\begin{align*}
\#_{2p}(x)&=\#_{2p}(\tr)=\#_{2p}(\fa)=0,\\
\#_{2p}(t(s))&=\#_{2p}(t)+\#_{2p}(s),\\
\#_{2p}(u(s,s'))&=\#_{2p}(u)+\#_{2p}(s)+\#_{2p}(s').
\end{align*}

Notice $T_{TND}(X)\subseteq_K T_C^{1,2}(X)\subseteq_K
T_C(X)$, where $M\subseteq_K N$ if for each term $t\in M$
there is a term $r\in N$ with $r=_K t$ for 
$K\in\{\fr,\rp,\con,\wm,\mem,\st\}$. We write
$\in_{K}$ for the membership relation
associated with $\subseteq_K$.

The sets $T_{TND}(A)$, $T_C^{1,2}(A)$ and $T_C(A)$ contain the 
closed substitution instances of the respective term sets
when constants from $A$ are substituted for the variables.
The set $T_C^{1,2}(A,X)$
contains the terms constructed from $T_C^{1,2}(A)$ 
and $T_C^{1,2}(X)$.
For given terms $r(x)\in T_C^{1,2}(A,\{x\})$ and
$t\in T_C^{1,2}(A)$ we write $r(t)$ for the term obtained 
by substituting $t$ for all $x$ in $r$. 
(Another common notation 
for $r(t)$ is $[t/x]r(x)$.)
We extend the definition of $\#_{2p}(t)$ to $T_C^{1,2}(A,X)$
in the expected way by defining $\#_{2p}(a)=0$ for all $a\in A$.

Clearly for all $K$,
\[T_{TND}(A)\subseteq_K T_C^{1,2}(A)\subseteq_K
T_C(A).\]
From Proposition~\ref{prop:fr} we find that 
$a\lef a\rig\neg a\not\in_{\fr}T_{TND}(A)$, thus
\[T_{TND}(A)\varsubsetneq_{\fr} T_C^{1,2}(A).\]
Theorem~\ref{thm:fr} below
establishes that $T_C^{1,2}(A)
\varsubsetneq_{\fr}T_C(A)$ as $a\lef b\rig c\not\in_{\fr}
T_C^{1,2}(A)$. This result transfers to $\rp$-congruence
without modification. However, in $\wm$-congruence we find
\[a\lef b\rig c =_{\wm}(\tr\lef b\rig c)\lef
(a\lef b\rig\tr)\rig \fa=(\neg b\leftor a)\leftand(b\leftor c),\]
thus $a\lef b\rig c\in_{\wm}
T_C^{1,2}(A)$.

\begin{theorem}
\label{thm:fr}
If $|A|>2$ then the conditional operator cannot be expressed 
modulo free valuation 
congruence in $T_C^{1,2}(X)$. 
\end{theorem}

\begin{proof}
It is sufficient to prove that 
$a\lef b\rig c\not\in_{\fr} T^{1,2}_C(A)$.

Towards a contradiction,
assume $t\in T_C^{1,2}(A)$ is a term such that 
$t=_{\fr}a\lef b\rig c$ and $\#_{2p}(t)$ is minimal
(i.e., if $u\in T_C^{1,2}(A)$ and $u=_{\fr}t$ then
$\#_{2p}(u)\geq \#_{2p}(t)$).

We first argue that $t\not\equiv f(b,t')$ for some
binary function $f$ and term $t'$. Suppose 
otherwise, then $b$ must be 
the central condition in $f(b,t')$, so 
$f(b,t')=_{\fr}g(b,t')\lef b\rig h(b,t')$ 
for certain binary functions $g$ and $h$ in 
$T_C^{1,2}(X)$. Because it is neither the case
that $b$ can occur as a central condition
in both $g(b,t')$ and 
in $h(b,t')$, nor that each of these can be modulo 
$\fr$ in $\{\tr,\fa\}$, we find
\[t=_{\fr}(P\lef t'\rig Q)\lef b\rig
(P'\lef t'\rig Q')\]
for certain $P,P',Q,Q'$. The
only possibilities 
left are that the central atom of $t'$ is either $a$ or $c$,
and both choices contradict $f(b,t')=_{\fr}a\lef b\rig c$.

So it must be the case that
\[t\equiv r(f(b,t'))\]
for some term $r(x)\in T_C^{1,2}(\{x\})$
such that $b$ is central in $f(b,t')$ and
$x$ is central in $r(x)$. 
If no such term $r(x)$ exists,
then $t\equiv f'(a')$ with $f'(x)$ a unary operator
definable in $T_C^{1,2}(\{x\})$ and $a'\in A$, 
which cannot hold because
$t$ needs to contain $a$, $b$ and $c$.

Also there cannot be a unary
function $f'\in T_C^{1,2}(\{x\})$ with 
$r(f'(b))=_{\fr}r(f(b,t'))$,
otherwise $r(f'(b))\in T^{1,2}_C(A)$ while
\[\#_{2p}(r)=\#_{2p}(r(f'(b)))
<\#_{2p}(r(f(b,t')))=\#_{2p}(r)+\#_{2p}(t')+1,\] 
which contradicts the minimality of $\#_{2p}(t)$.

As $x$ is central in $f(x,y)$ we may write
\[f(x,y)=_{\fr}g(x,y)\lef x\rig h(x,y)\]
for certain binary functions $g$ and $h$ in 
$T_C^{1,2}(X)$. Because
$b$ is central in $t$ we find
\[t=_{\fr}
r(g(b,t')\lef b\rig h(b,t')).
\]

We proceed with a case distinction on the form that 
$g(b,t')$ and $h(b,t')$
may take. At least one of these is modulo $\fr$
not equal to $\tr$ or $\fa$ (otherwise $f(b,t')$ 
could be replaced by $f'(b)$ for some unary function $f'$ 
and this was excluded above).

\begin{enumerate}

\item
Suppose $g(b,t')\not\in_{\fr}\{\tr,\fa\}$ and 
$h(b,t')\not\in_{\fr}\{\tr,\fa\}$. First notice
that $b$ cannot occur as a central condition
in both $g(b,t')$ and 
in $h(b,t')$. So, both $g(x,y)$ and $h(x,y)$ can
be written as a conditional composition 
with $y$ as the central variable, and we find
\[t=_{\fr}r((P\lef t'\rig Q)\lef b\rig
(P'\lef t'\rig Q'))\]
for certain closed terms $P,Q,P',Q'$. 
By supposition $t'\not\in_{\fr}
\{\tr,\fa\}$, and the only possibilities 
left are that its central atom equals both $a$ and $c$,
which clearly is impossible.

\item We are left with four cases:
either $a$ is central in $g(b,t')$ and 
$h(b,t')\in_{\fr}\{\tr,\fa\}$, or
$c$ is central in $h(b,t')$ and 
$g(b,t')\in_{\fr}\{\tr,\fa\}$.
These cases are symmetric and it suffices to consider
only the first one, the others can be dealt with similarly. 

So assume $a$ is central in $g(b,t')$ and 
$h(b,t')=_{\fr}\tr$, hence
\[g(b,t')=_{\fr}
P\lef a\rig Q\quad\text{ for some $P,Q\in\{\tr,\fa\}$}.
\]
We find
\[t=_{\fr}r((P\lef a\rig Q)\lef b\rig\tr),\]
and we distinguish two cases:

$i$. $P\equiv \tr$ or $Q\equiv \tr$. 
Now a central $c$ can be 
reached after a negative reply to $b$.
But this central $c$ can also be reached
after a positive reply to $b$ and the appropriate reply 
to $a$, which contradicts free congruence with
$a\lef b\rig c$.

$ii$. $P\equiv Q\equiv\fa$. Then the reply to $a$ 
in $r((\fa\lef a\rig \fa)\lef b\rig\tr)$ is not used,
which also contradicts free congruence with
$a\lef b\rig c$.
\end{enumerate}
This concludes our proof.
\end{proof}

We will now argue that $a\lef b\rig c\not\in_{\con}
T_C^{1,2}(A)$. We will make use of additional
operators $\tr_a$ and $\fa_a$ for each atom $a\in A$, 
defined for all $b\in A$ and terms $t,r\in T_C^{1,2}(A)$ by
\begin{align*}
\tr_a(\tr)&=\tr,&\fa_a(\tr)&=\fa,\\
\tr_a(\fa)&=\tr,&\fa_a(\fa)&=\fa,\\
\tr_a(t\lef b\rig r)&=t\lef b\rig r\text{~~if }a\neq b,
&\fa_a(t\lef b\rig r)&=t\lef b\rig r\text{~~if }a\neq b,\\
\tr_a(t\lef a\rig r)&=\tr_a(t),&\fa_a(t\lef a\rig r)&=\fa_a(r).
\end{align*}
Observe that $\tr_a$ ($\fa_a$) simplifies a term $t$
as if it is a subterm of $a\andthen t$ with the additional
knowledge that the reply on $a$ has been \tr.
We notice that 
\[t\lef a \rig r=_{\con}\tr_a(t)\lef a \rig \fa_a(r).\]

We define a term $P$ to have the property 
$\phi_{a,b,c}$ if
\begin{itemize}
\item
the central atom of $\tr_b(P)$ equals $a$,
$\,\tr_a(\tr_b(P))\in_{\con}\{\tr,\fa\}$ and
$\fa_a(\tr_b(P))\in_{\con}\{\tr,\fa\}$, and
$\tr_a(\tr_b(P))\ne_{\con}
\fa_a(\tr_b(P))$,
\item
the central atom of $\fa_b(P)$ equals $c$,
$\,\tr_c(\tr_b(P))\in_{\con}\{\tr,\fa\}$ and
$\fa_c(\tr_b(P))\in_{\con}\{\tr,\fa\}$, and
$\tr_c(\tr_b(P))\ne_{\con}
\fa_c(\tr_b(P))$.
\end{itemize}
Typically, $a\lef b\rig c$ has property $\phi_{a,b,c}$.

\begin{theorem}
\label{THM:CR}
If $|A|>2$ then the conditional operator cannot be 
expressed modulo contractive valuation 
congruence in $T_C^{1,2}(X)$.
\end{theorem}

\begin{proof}
Let $a,b,c\in A$. It is sufficient to show that no term
in $T^{1,2}_C(A)$ has property $\phi_{a,b,c}$. 
A detailed proof of this fact is 
included in Appendix~\ref{app:A}.
\end{proof}

Finally, we observe that $x\lef y\rig z$
is expressible in $\CP_{\mem}$ using $\leftand$ and $\neg$
only: first $\leftor$ is expressible, and
\begin{align*}
\CP_{\mem}\vdash
(y\leftand x)\leftor(\neg y\leftand z)
&=\tr\lef(x\lef y\rig\fa)\rig(z\lef(\fa\lef y\rig\tr)\rig\fa)\\
&=\tr\lef(x\lef y\rig\fa)\rig(\fa\lef y\rig z)\\
&=(\tr\lef x\rig(\fa\lef y\rig z))\lef y\rig(\fa\lef y\rig z)\\
&\stackrel{\eqref{mini1}}=(\tr\lef x\rig\fa)\lef y\rig(\fa\lef y\rig z)\\
&\stackrel{\eqref{mini3}}=x\lef y\rig z.
\end{align*}
Thus, for $x,y,z\in X$ it holds that
$(x\lef y\rig z)\in_{\mem}T^{1,2}_C(X)$. We leave
it as an open question whether 
$(x\lef y\rig z)\in_{\wm}T^{1,2}_C(X)$.

\section{Projections and the projective limit model}
\label{sec:proj}
In this section we introduce the 
\emph{projective limit model} $\PLM$ for defining
potentially infinite propositional statements.

Let $\Prop$ be the domain of the initial algebra of \CP,
so each element in \Prop\ can be represented by a basic
form. 
Let $\Nplus$ denote $\Nat\setminus \{0\}$.
We first define a so-called \emph{projection} operator
\[\pi:\Nplus\times \Prop\rightarrow \Prop,\]
which will be used to finitely approximate every propositional
statement in $\Prop$. 
We further write 
\[\pi_n(P)
\]
instead of $\pi(n,P)$.
The defining equations for the $\pi_n$-operators are these
($n\in\Nplus$):
\begin{align}
\pi_{n}(\tr)&=\tr,\\
\pi_{n}(\fa)&=\fa,\\
\pi_1(x\lef a\rig y)&=a,\\
\pi_{n+1}(x\lef a\rig y)&=\pi_n(x)\lef a \rig\pi_n(y),
\end{align}
for all $a\in A$. We write $\PR$ for this set 
of equations.

We state without proof that 
$\CP+\PR$ is a conservative extension of \CP\ and 
mention the following derivable identities in $\CP+\PR$
for $a\in A$ and $n\in\Nplus$:
\begin{align*}
\pi_{n}(a)&=\pi_{n}(\tr\lef a\rig\fa)=a,
\\
\pi_{n+1}(a\andthen x)&=a\andthen \pi_n(x).
\end{align*}
Below we prove that for each propositional statement $P$ there exists $n
\in\Nplus$ such that for all $j\in\Nat$,
\[\pi_{n+j}(P)=P.\]
 
The following lemma 
establishes how the arguments of the 
projection of a conditional composition can be 
restricted to certain projections, in particular 
\begin{equation}
\label{projc}
\pi_n(P\lef Q\rig R)=
\pi_n(\pi_{n}(P)\lef \pi_{n}(Q)\rig \pi_{n}(R)),
\end{equation}
which is a property that we will use in the definition
of our projective limit model.

\begin{lemma} 
\label{proj}
For all $P,Q,R\in\Prop$ and all $n\in\Nplus$, 
$k,\ell,m\in\Nat$,
\[
\pi_n(P\lef Q\rig R)=
\pi_n(\pi_{n+k}(P)\lef \pi_{n+\ell}(Q)\rig \pi_{n+m}(R)).
\]
\end{lemma}

\begin{proof}
We may assume that $Q$ is a
basic form. We apply structural induction on $Q$.
\\[2mm]
If 
$Q=\tr$ then we have to prove that 
for all $n\in\Nplus$ and
$k\in\Nat$,
\[
\pi_n(P)=
\pi_n(\pi_{n+k}(P)).
\]
We may assume that $P$ is a basic form and we apply
structural induction on $P$. If $P\in\{\tr,\fa\}$ we
are done. If $P= P_1\lef a\rig P_2$ then we proceed by
induction on $n$. The case $n=1$ is trivial,
and 
\begin{align*}
\pi_{n+1}(P)&
=\pi_{n+1}(P_1\lef a \rig P_2)\\
&=\pi_{n}(P_1)\lef a \rig \pi_n(P_2)\\
&\stackrel{IH}=\pi_{n}(\pi_{n+k}(P_1))\lef a\rig 
  \pi_n(\pi_{n+k}(P_2))\\
&=\pi_{n+1}(\pi_{n+k}(P_1)\lef  
  a \rig \pi_{n+k}(P_2))\\
  &=\pi_{n+1}(\pi_{n+k+1}(P)).
\end{align*}
\\[2mm]
If $Q=\fa$: similar. 
\\[2mm]
If $Q=Q_1\lef a \rig Q_2$
then we proceed by induction on $n$. The case $n=1$ is trivial,
and 
\\\\
$
\begin{array}{l}
\pi_{n+1}(P\lef Q\rig R)\\
{}=\pi_{n+1}(P\lef (Q_1\lef a \rig Q_2)\rig R)\\
{}=\pi_{n+1}((P\lef Q_1\rig R)\lef a \rig (P\lef Q_2\rig R))\\
{}=\pi_{n}(P\lef Q_1\rig R)\lef a \rig \pi_n(P\lef Q_2\rig R)\\
\stackrel{IH}=\pi_{n}(\pi_{n+k}(P)\lef \pi_{n+\ell}(Q_1)\rig 
  \pi_{n+m}(R))\lef a \rig \pi_n(\pi_{n+k}(P)\lef
  \pi_{n+\ell}(Q_2)\rig \pi_{n+m}(R))\\
{}=\pi_{n}(\pi_{n+k+1}(P)\lef \pi_{n+\ell}(Q_1)\rig 
  \pi_{n+m+1}(R))\lef a\rig \pi_n(\pi_{n+k+1}(P)\lef
  \pi_{n+\ell}(Q_2)\rig \pi_{n+m+1}(R))\\
{}=\\\pi_{n+1}((\pi_{n+k+1}(P)\lef \pi_{n+\ell}(Q_1)\rig
  \pi_{n+m+1}(R))\lef a\rig(\pi_{n+k+1}(P)\lef
  \pi_{n+\ell}(Q_2)\rig \pi_{n+m+1}(R)))\\
{}=\pi_{n+1}(\pi_{n+k+1}(P)\lef (\pi_{n+\ell}(Q_1)\lef
  a \rig \pi_{n+\ell}(Q_2))\rig \pi_{n+m+1}(R))\\
{}=\pi_{n+1}(\pi_{n+1+k}(P)\lef \pi_{n+1+\ell}(Q)\rig \pi_{n+1+m}(R)).
\end{array}
\\
$
\end{proof}

The \emph{projective limit model} $\PLM$ is
defined as follows:

\begin{itemize}
\item The domain of $\PLM$ is the
set of \emph{projective sequences}
$(P_n)_{n\in\Nplus}$: these are all sequences with the
property that all $P_n$ are in $\Prop$ and
satisfy
\[\pi_{n}(P_{n+1})=P_n,\] 
so that they can be seen as successive projections of
the same infinite propositional statement 
(observe that $\pi_n(P_n)=P_n$). 
We write $\Propinfty$ for this domain, and we further write 
$(P_n)_n$ instead of $(P_n)_{n\in\Nplus}$.

\item Equivalence of projective sequences in $\PLM$ 
is defined component-wise, thus 
$(P_n)_{n}=(Q_n)_{n}$ if for 
all $n$, $P_n=Q_n$.
\item
The constants \tr\ and \fa\ are interpreted in $\PLM$
as the projective sequences that consist solely of these 
respective constants.
\item
An atomic proposition $a$ is interpreted in $\PLM$ 
as the projective sequence $(a,~a,~a,...)$.
\item Projection in $\PLM$ is defined component-wise, thus 
$\pi_k((P_n)_{n})=(\pi_k(P_n))_{n}$.
\item
Conditional composition in $\PLM$ is defined using projections:
\[(P_n)_{n}\lef(Q_n)_{n}\rig(R_n)_{n}=
(\pi_n(P_n\lef Q_n\rig R_n))_{n}.\]
The projections are needed if the depth of
a component $P_n\lef Q_n\rig R_n$ exceeds $n$. 
equation~\eqref{projc} implies that this definition
indeed yields a projective sequence:
\begin{align*}
\pi_n(\pi_{n+1}(P_{n+1}\lef Q_{n+1}\rig R_{n+1}))
&=\pi_n(P_{n+1}\lef Q_{n+1}\rig R_{n+1})\\
&=\pi_n(\pi_n(P_{n+1})\lef \pi_{n}(Q_{n+1})\rig \pi_n(R_{n+1}))\\
&=\pi_n(P_n\lef Q_n\rig R_n).
\end{align*}

\end{itemize}

The following result can be proved straightforwardly:
\begin{theorem}
$\PLM\models\CP + \PR$.
\end{theorem}

The projective limit
model $\PLM$ contains elements that are not the 
interpretation of finite propositional statements 
in $\Prop$ (in other words, 
elements of infinite depth).
In the next section we discuss some examples.

\section{Recursive specifications}
\label{sec:recspec}
In this section we discuss 
\emph{recursive specifications} over $\Sigma_\CP(A)$, 
which provide an alternative and simple way
to define propositional statements in $\PLM$. 
We first restrict ourselves to a 
simple class of recursive specifications: Given $\ell>0$, a set
\[E=\{X_i=t_i\mid i=1,...,\ell\}\]
of equations is
a \emph{linear specification} over $\Sigma_\CP(A)$ if
\[t_i::=\tr\mid\fa\mid X_j\lef a_i\rig X_k\]
for $i,j,k\in\{1,...,\ell\}$ and $a_i\in A$. 
A \emph{solution} for $E$ 
in $\PLM$ is a series of propositional statements 
\[(P_{1,n})_n, ...,(P_{\ell,n})_n\]
such that $(P_{i,n})_n$ solves the equation for $X_i$.
In $\PLM$, solutions for linear specifications exist. 
This follows from the property
that for each $m\in\Nplus$, $\pi_m(X_i)$ can be computed as 
a propositional statement
in $\Prop$ by replacing variables $X_j$ by $t_j$
sufficiently often. 
For example, if 
\[E=\{X_1=X_3\lef a\rig X_2,~ X_2=b\andthen X_1,
~X_3=\tr\}\]
we find $\pi_m(X_3)=\pi_m(\tr)=\tr$ for all $m\in\Nplus$,
and
 \begin{align*}
 \pi_1(X_2)&=\pi_1(b\andthen X_1)
 &\pi_{m+1}(X_2)&= \pi_{m+1}(b\andthen X_1)\\
 &=b,&&=b\andthen \pi_m(X_1),\\[2mm]
 \pi_1(X_1)
 &=\pi_1(X_3\lef a\rig X_2)
 &\pi_{m+1}(X_1)
 &=\pi_{m+1}(X_3\lef a\rig X_2)\\
 &=a,
 &&=\tr\lef a\rig\pi_{m}(X_2),
 \end{align*}
and we can in this way construct a projective sequence
per variable. We state without proof
that for a linear specification 
$E=\{X_i=t_i\mid i=1,...,\ell\}$
such sequences model \emph{unique} solutions in 
$\PLM$,\footnote{$\PLM$ can be turned into a metric space by
 defining $d((P_n)_n,(Q_n)_n)=2^{-n}$ for $n$ the least value
 with $P_n\neq Q_n$. The existence of unique solutions for
 linear specifications then follows from Banach's fixed point
 theorem; a comparable and detailed account of this fact
 can be found in~\cite{Vu08}.} 
and  we write
\[\langle X_i | E\rangle\]
for the solution of $X_i$ as defined in $E$. In order to
reason about linearly specified propositional statements, 
we add
these constants to the signature $\Sigma_{\CP}$, 
which consequently satisfy the equations 
\[\langle X_i | E\rangle= \langle t_i | E\rangle\]
where $\langle t_i | E\rangle$ is defined by replacing 
each $X_j$ in $t_i$ by $\langle X_j | E\rangle$. The proof
principle introducing these identities is called the 
\emph{Recursive Definition Principle} (RDP), and for linear
specifications RDP is valid in the projective
limit model $\PLM$.\footnote{A nice and comparable account 
of the validity of RDP in the 
projective limit model for \ACP\ is given in~\cite{BW90}.
In that text book, a sharp distinction is made between RDP|stating
that certain recursive specifications have \emph{at least} a solution
per variable|and the Recursive Specification Principle (RSP),
stating that they have \emph{at most} one solution per variable.
The uniqueness of solutions per variable then
follows by establishing the validity of both RDP and RSP.}
As illustrated above, all solutions satisfy
\[\langle X_i | E\rangle=(\pi_n(\langle X_i | E\rangle))_n.\]

Some examples of propositional statements defined by recursive
specifications are these:
\begin{itemize}
\item
For $E=\{X_1=X_2\lef a\rig X_3,~ X_2=\tr,~ X_3=\fa\}$
we find 
 \[\langle X_1 | E\rangle=(a,~a,~a,~...)\] 
which in the projective limit model represents the 
atomic proposition $a$. Indeed, by RDP 
we find $\langle X_1 | E\rangle=
\langle X_2 | E\rangle\lef a \rig\langle X_3 | E\rangle=
\tr \lef a \rig\fa=a$.

\item
For $E=\{X_1=X_2\lef a\rig X_3,~ X_2=\tr,~ X_3=\tr\}$
we find 
\[\langle X_1 | E\rangle=(a,~
a\andthen\tr,~a\andthen\tr,~a\andthen\tr,~...)\] 
which in the projective limit model represents $a\andthen\tr$.
By RDP we find $\langle X_1 | E\rangle=a\andthen\tr$.

\item
For $E=\{X_1=X_3\lef a\rig X_2,~ X_2=b\andthen X_1,
~X_3=\tr\}$ as discussed above,
we find 
 \[\langle X_1 | E\rangle=(a,~\tr\lef a\rig b,~
\tr\lef a\rig b\andthen a,~
\tr\lef a\rig b\andthen (\tr\lef a\rig b),~...)\] 
which in the projective limit model represents an 
\emph{infinite}  propositional statement, that is, one that satisfies 
\[\pi_i(\langle X_1 | E\rangle)=\pi_j(\langle X_1 | E\rangle)
\quad\Rightarrow\quad i=j,\] 
and thus has infinite depth.
By RDP we find 
$\langle X_1 | E\rangle=\tr\lef a\rig b\andthen\langle X_1 | E\rangle$.
 We note that the infinite propositional statement 
 $\langle X_1 | E\rangle$
can be characterized as \[\textit{while $\neg a$ do $b$}.\]
\end{itemize}

An example of a projective sequence that cannot
be defined by a linear specification, but that can be defined
by the \emph{infinite} linear specification 
$I=\{X_{i}=t_{i}\mid i\in\Nplus\}$ with 
\[
t_i=
\begin{cases}
a\andthen X_{i+1}&\text{if $i$ is prime},\\
b\andthen X_{i+1}&\text{otherwise},
\end{cases}
\]
is $\langle X_1 | I\rangle$, satisfying
\[
\langle X_1 | I\rangle=
(b,~b\andthen a,~b\andthen a\andthen a,
~b\andthen a\andthen a\andthen b,
~b\andthen a\andthen a\andthen b\andthen a,~...).
\]
Other examples of projective sequences that cannot
be defined by a finite linear specification
are
$\langle X_j | I\rangle$ for any $j>1$.

Returning to Example~\eqref{pedes} of a propositional statement
sketched in Section~\ref{sec:motiv},
we can be more explicit now: the  recursively defined 
propositional statement
$\langle X_1|E\rangle$ with $E$ containing
\begin{align*}
X_1&=X_2\lef\textit{green-light}\rig X_1,\\
X_2&=X_3\lef(\textit{look-left-and-check}\leftand 
\textit{look-right-and-check}\leftand\textit{look-left-and-check})\rig X_1,\\
X_3&= ...
\end{align*}
models in a straightforward way a slightly larger part 
of the processing of a pedestrian planning to
cross a road with two-way traffic driving on the right.

\section{Application perspective}
\label{sec:app}
Although we consider the family of models for \CP,
ranging from free valuation congruence to static valuation
congruence to be of independent importance, it is
evidently reasonable to ask what application perspective these 
matters may have.

There is a remarkably wide range of scenarios of usage for
propositional statements in general. A proposition
$P$ may be used to express a matter of fact, a belief, an 
objective or a desire. It may also express a general law
considered invariant in time in an appropriate setting,
or rather an (intended) variant for some product of
human design. Then $P$ may serve as an element of a 
knowledge-base, or as a phrase used in communication
or broadcast. Finally and most relevant to the discussion 
below a major use for a
proposition is to serve as a condition which impacts 
future behavior. This is primarily exemplified in 
program fragments of the form
\[...;\textit{if $\{P\}$ then $\{S_1\}$ else $\{S_2\}$};...\]
During execution of this fragment, $P$ must be evaluated.
Evaluation of $P$ can take many forms, ranging from 
finding a proof for either $P$ or $\neg P$ from the
information contained in some knowledge-base, by means of
either monotonic or non-monotonic logic, perhaps
constrained by resource bounds, to bottom-up evaluation 
using the atoms contained in $P$ as primitive queries.
The primary role of $P$ as a propositional statement 
with constants from 
$A\cup\{\tr,\fa\}$ and connectives 
$\neg,\leftand,\leftor,\_\lef\_\rig\_\:$
is to serve as a condition in the latter sense, as $P$
expresses a meaningful condition while its particular
form in addition has an algorithmic significance
by imposing a specific strategy for sequential 
bottom-up evaluation.

Let $P=a\leftand((b\leftor(\neg a\leftand c))\leftor
(d\leftand \neg a))$. Then evaluation of $P$ will
involve one, two or even three evaluations of $a$. To begin with
we provide a survey of different ways in which the
atomic query $a$ can be handled:
\begin{enumerate}
\item The atomic query $a$ may inspect a static database. 
Subsequent queries
provide identical results, different queries don't affect
one another.

\item The atomic query $a$ may inspect a dynamic database. 
Subsequent queries
may return different values but will not cause fluctuations
in the response of other atomic queries.

\item The atomic query 
$a$ may be handed over to a ``truth maintenance system''
(TMS) which tries to prove it, and otherwise returns \fa.
The knowledge-base managed by this TMS itself may be regularly
changed by means of a mechanism (often called a ``belief revision 
mechanism'') that processes a stream of incoming, potentially
authoritative information.

\item The atomic query 
$a$ itself may call another program that has some or
even significant side effects which may influence the
replies provided for forthcoming atomic tests.

\item Like 4, but observing the side effect of a task
may be limited to agents that operate on a high security level.
\end{enumerate}

We say that a propositional statement $P$ is in \emph{monotest 
form} if no atom can be evaluated more than once.
If $P$ is in {monotest  form} then
fluctuations regarding the evaluation of a single test 
have no impact on $P$. Rewriting to monotest form can always
be done modulo static valuation congruence. The operator
$\cach(X)$ on basic forms will produce an equivalent
monotest form:
\begin{align*}
\cach(\tr)&=\tr,\quad\cach(\fa)=\fa\\
\cach(x\lef a\rig y)&=\cach([\tr/a]x)\lef a\rig\cach([\fa/a]y).
\end{align*}
Now transforming a propositional statement $P$ into
$\cach(bf(P))$ may involve a combinatorial explosion in size.
Suppose $\mf(P)$ finds a monotest form for $P$ modulo $=_\mem$
in polynomial time. Then
\[\SAT_\mem(P)\iff \SAT_\mem(\mf(P))\iff \SAT_\con(\mf(P)).\]
Then $\SAT_\mem$ would be in P while it is known to be NP-complete.
As it turns out, the combinatorial explosion in size that comes with
the transformation $P\mapsto \cach(bf(P))$ is no coincidence
and for that reason for larger conditions
it is reasonable to assume that 
these are not in monotest form.

Except for the case that atomic tests are evaluated from a
static database, fluctuating replies 
cannot be excluded. Now consider program $X$ with
\[X=C[\text{if $\{P\}$ then $\{S_1\}$ else $\{S_2\}$}]\]
and assume that $P$ is not in monotest form and that
subsequent evaluations of atomic tests may have different results.
At this stage awareness of proposition algebra may be of use
in contemplating revisions of the design of the program.
If the ``internal logic'' of $P$ (i.e., the rationale
of $P$'s occurrence in $X$) is defeated by such fluctuations,
\[C[\text{if $\{\cach(bf(P))\}$ then $\{S_1\}$ else $\{S_2\}$}]\]
may be considered an improved design of the program (to implement 
this concisely, one may need auxiliary Boolean variables).
But a disadvantage of this transformation is that in some cases
a cached outcome of an atomic test is ``outdated''. If neither
consistency (no fluctuation), nor the use of outdated values 
is a serious worry, there is no incentive to change the design of $X$.
If, however, the use of of outdated values is to be preferably
prevented, then one may instead decide to restart the evaluation 
of $P$ whenever a repeated atomic test returns a different value.
Here is an operation $\ree_{V,W}(Q)$ which serves that purpose,
where $V,W\subseteq A$, $V\cap W=\emptyset$  and $Q$ is a 
subformula of $P$.
$\ree_{V,W}(Q)$ evaluates $Q$ using the information that
preceding tests $a\in V$ have returned \tr\ and preceding tests
$a\in W$ have returned \fa. If a reply is observed that fails 
to match with this information, evaluation of
$P$ starts again with both $V$ and $W$ made empty:
\begin{align*}
\ree(P)&=\ree_{\emptyset,\emptyset}(P),\\
\ree_{V,W}(\tr)&=\tr,\quad\ree_{V,W}(\fa)=\fa,\\
\ree_{V,W}(Q\lef a\rig R)&=\ree_{V\cup\{a\},W}(Q)\lef a\rig
\ree_{V,W\cup\{a\}}(R)\quad\text{if $a\not\in V\cup W$},\\
\ree_{V,W}(Q\lef a\rig R)&=\ree_{V,W}(Q)\lef a\rig
\ree_{\emptyset,\emptyset}(P)\quad\text{if $a\in V$},\\
\ree_{V,W}(Q\lef a\rig R)&=\ree_{\emptyset,\emptyset}(P)\lef a\rig
\ree_{V,W}(R)\quad\text{if $a\in W$}.
\end{align*}
Observe that $\ree(P)$ determines a proposition in the projective
limit model.

Now one may be afraid that repeated fluctuations cause $\ree(P)$ 
to diverge in adverse circumstances. A new test
$DLNI$ (deadline not immanent) may be introduced which is used 
to assert whether or not time suffices for a full re-evaluation of $P$.
The following modification for the design of $\ree$ is plausible:
\begin{align*}
\ree_{V,W}(Q\lef a\rig R)&=\ree_{V,W}(Q)\lef a\rig
(\ree_{\emptyset,\emptyset}(P)\lef DLNI\rig R)\quad\text{if $a\in V$},
\end{align*}
and symmetrically for $a\in W$. Now one might be dissatisfied with
$R$ not using the information contained in $V$ and $W$. 
Then one may use instead $[\tr/V,\fa/W]R$, the modification of $R$
in which its atoms that are in $V$ ($W$) are
substituted by $\tr~(\fa)$:
\[
\ree_{V,W}(Q\lef a\rig R)=\ree_{V,W}(Q)\lef a\rig
(\ree_{\emptyset,\emptyset}(P)\lef DLNI\rig [\tr/V,\fa/W]R)\\
\text{ if $ a\in V$},
\]
and symmetrically for $a\in W$.

Proposition algebra provides an approach that allows to compare 
and develop these design alternatives within a formal setting. 
The advantage of doing so is a matter of separation of concerns,
which can be considered a big issue for imperative programming.

\section{Conclusions}
\label{sec:conc}

Proposition algebra in the form of \CP\ for  
propositional statements with conditional composition and
either enriched or not
with negation and sequential connectives, is proposed as an 
abstract data type. 
Free valuations provide the natural semantics 
for \CP\ and 
these are semantically at the opposite end of static valuations.
It is shown that taking conditional composition and 
free valuations as
a point of departure implies that a ternary connective is needed 
for functional completeness; binary connectives are not sufficient.
Furthermore, \CP\ admits a meaningful and 
non-trivial extension to projective limits, and this 
constitutes the most simple case of an inverse limit construction 
that we can think of.

The potential role of proposition algebra is only touched upon by some 
examples. It remains a challenge to find convincing examples that
require reactive valuations, and to find earlier accounts of this 
type of semantics for propositional logic.
The basic idea of proposition algebra with 
free and reactive valuations
can be seen as
a combination of the following two ideas:
\begin{itemize}
\item Consider atomic propositions as events (queries) that can have a 
side effect in a sequential system, and take McCarthy's sequential 
evaluation as described in~\cite{McC63} to 2-valued 
propositional logic; this motivates reactive valuations as those
that define evaluation or computation as a sequential phenomenon.
\item
In the resulting setting, Hoare's conditional composition as
introduced in~\cite{Hoa85} is more
natural than the sequential, non-commutative versions of 
conjunction and disjunction, and (as it appears)
more expressive: a ternary connective is needed anyhow.
\end{itemize}
For conditional composition we have chosen for the notation 
\[\_\lef \_ \rig \_\]
from Hoare~\cite{Hoa85} in spite of the 
fact that our theme is technically closer to thread 
algebra~\cite{BM07} where a different notation is used.
We chose for the notation $\_\lef\_\rig\_\:$
because its most well-known 
semantics is static valuation semantics (which is simply 
conventional propositional logic) for which this notation 
was introduced in~\cite{Hoa85}.\footnote{This notation 
 was used by Hoare in his 1985~book on CSP~\cite{Hoa85a} 
 and by Hoare \emph{et al.} in 
 the well-known 1987 paper \emph{Laws of Programming}~\cite{HHH87} 
 for expressions $P\lef b\rig Q$ with $P$ and $Q$ programs and 
 $b$ a Boolean expression without mention of~\cite{Hoa85}
 that appeared in 1985.}
To some extent, thread algebra and propositional logic in the 
style of~\cite{Hoa85} are models of the same signature.
A much more involved use of
conditional composition can be found in~\cite{PZ07}, where the
propositional fragment of Belnap's four-valued logic~\cite{Bel77} 
is characterized
using only conditional composition and his four constants
representing these truth values.

In this paper we assumed that $|A|>1$. The case that $|A|=1$ is in detail
described in~\cite{Chris}. In particular, $=_{\rp}$ and $=_{\st}$ and thus 
all valuation congruences in between coincide in this case.

\paragraph{Related work.}
We end with a few notes on related matters.
\begin{enumerate}

\item On McCarthy's conditional expressions~\cite{McC63}.
In quite a few papers 
the `lazy evaluation' semantics proposed in 
McCarthy's work is discussed, or taken
as a point of departure. We mention a few of these
in reverse chronological order:

H\"ahnle states in his paper \emph{Many-valued logic, 
partiality, and abstraction in formal specification 
languages}~\cite{Hahn05} that 
\begin{quote}
\emph{``sequential conjunction
[...] represents the idea that
if the truth value can be determined after evaluation of the first 
argument, then the result is computed without looking at the 
second argument. Many programming languages contain operators that
exhibit this kind of behavior"}. 
\end{quote}
Furthermore, 
Konikowska describes in~\cite{Kon96} a model of so-called
McCarthy algebras in terms of three-valued logic, while 
restricting to the well-known symmetric binary connectives, and 
provides sound axiomatizations and representation results.
This is achieved by admitting only \tr\ and \fa\ as constants
in a McCarthy algebra, and distinguishing an element $a$ as in
one of four possible classes (`positive' if $a\vee x=a$, 
`negative' if $a\wedge x=a$, `defined' if $a\wedge \neg a=\fa$,
and `strictly undefined' if $a=\neg a$).

Finally, Bloom and Tindell discuss in their paper
\emph{Varieties of ``if-then-else''}~\cite{BT83} various
modelings of conditional composition, both 
with and without a truth value \emph{undefined},
while restricting to the ``redundancy law'' 
\[(x\lef y\rig z)\lef y\rig u=x\lef y \rig u,\]
a law that we called \axname{CPcontr} in Section~\ref{sec:com}
and that
generalizes the axiomatization of contractive 
valuation congruence defined in that section to an extent in which
only the difference beteen \tr, \fa\ and 
\emph{undefined} plays a decisive role.

As far as we can see, none of the papers mentioned here even 
suggests the idea of \emph{free} or \emph{reactive}
valuation semantics. Another example where sequential 
operators play a role is \emph{Quantum logic} (for a 
brief overview see
\cite{Mit04}), where next to normal conjunction a notion
of \emph{sequential conjunction} $\sqcap$ is exploited
that is very similar to $\leftand$ (and that despite
its notation is certainly not symmetric).
\item
Concerning projections and the projective limit model $\PLM$
we mention that in much current research and exposition,
projections are defined also for depth 0 
(see, e.g.,~\cite{BM07,Vu08} for thread algebra, 
and~\cite{Fok00} for process algebra). 
However, \CP\ does 
not have a natural candidate for $\pi_0(P)$ and therefore we
stick to the original approach as described in~\cite{BK84}
(and overviewed in~\cite{BW90}) that starts from projections
with depth 1.

\item
Reactive valuations were in a different form employed in 
accounts of process algebra with propositional statements: 
in terms of
operational semantics, this involves transitions
\[P\step{a,w} Q\]
for process expressions $P$ and $Q$
with $a$ an action and $w$ ranging over a class of valuations.
In particular this approach deals with 
process expressions that contain propositional statements 
in the
form of guarded commands, such as $\phi:\rightarrow P$ that
has a transition 
\[(\phi:\rightarrow P)\step{a,w}Q\]
if $P\step{a,w} Q$ and $w(\phi)=\tr$.
For more information about this approach, see, 
e.g.,~\cite{BP98,BP98a}.
\end{enumerate}

\appendix
\section{Proof of Theorem~\ref{THM:CR}}
\label{app:A}
\begin{proof}
This proof
has the same structure as the proof of Theorem~\ref{thm:fr},
but a few cases require more elaboration.

Towards a contradiction,
assume that $t\in
T_C^{1,2}(A)$ is a term with property $\phi_{a,b,c}$
and $\#_{2p}(t)$ is minimal.

We first argue that $t\not\equiv f(b,t')$ for some
binary function $f$ and term $t'$. Suppose 
otherwise, then $b$ must be 
the central condition in $f(b,t')$, so 
$f(b,t')=_{\con}g(b,t')\lef b\rig h(b,t')$ 
for certain binary functions $g$ and $h$ in 
$T_C^{1,2}(X)$. Notice
that because $b$ is not central in $\tr_b(g(b,t'))$,
a different atom must be central in this term, and this
atom must be $a$. For this to hold, $a$ must be central
in $\tr_b(t')$ and no atom different from $a$
can be tested by the first requirement of $\phi_{a,b,c}$.
So, after contraction of all further $a$'s we find
\[\tr_b(t')=_{\con}P\lef a \rig Q\]
with $P,Q\in\{\tr,\fa\}$, and similary 
\[\fa_b(t')=_{\con}P'\lef c \rig Q'\]
with $P',Q'\in\{\tr,\fa\}$.
If $P\not\equiv Q$ and $P'\not\equiv Q'$, 
then $t'$ is a 
term that satisfies $\phi_{a,b,c}$, but $t'$ 
is a term with lower $\#_{2p}$-value than 
$g(b,t')\lef b\rig h(b,t')$, which is a 
contradiction.
If either $P\equiv Q$ or
$P'\equiv Q'$, then
\[t=_{\con}(P\lef a\rig Q)\lef b\rig(P'\lef c\rig Q'),\]
which contradicts 
$\phi_{a,b,c}$.

So it must be the case that 
\[t\equiv r(f(b,t'))\]
for some term $r(x)\in T_C^{1,2}(\{x\})$
such that $b$ is central in $f(b,t')$
and $x$ is central in $r(x)$. 
If no such such term $r(x)$ exists,
then $t\equiv f'(a')$ with $f'(x)$ a unary operator
definable in $T_C^{1,2}(\{x\})$ and $a'\in A$, 
which cannot hold because
$t$ needs to contain $a$, $b$ and $c$.

Also there cannot be a unary
function $f'\in T_C^{1,2}(\{x\})$ with 
$r(f'(b))=_{\con}r(f(b,t'))$,
otherwise $r(f'(b))\in T^{1,2}_C(A)$ while
$\#_{2p}(r(f'(b)))<\#_{2p}(r(f(b,t')))$, which is 
a contradiction.

As $x$ is central in $f(x,y)$ we may write
\[f(x,y)=_{\con}g(x,y)\lef x\rig h(x,y)\]
for binary operators $g$ and $h$. Because
$b$ is central in $t$ we find
\[t=_{\con}
r\big (\tr_b(g(b,t'))\lef b\rig \fa_b(h(b,t'))\big ).\]

We proceed with a case distinction on the form that 
$\tr_b(g(b,t'))$ and $\fa_b(h(b,t'))$
may take. At least one of these is modulo $\con$
not equal to $\tr$ or $\fa$ (otherwise $f(b,t')$ 
could be replaced by $f'(b)$ for some unary function $f'$ 
and this was excluded above).

\begin{enumerate}

\item
Suppose $\tr_b(g(b,t'))\not\in_{\con}\{\tr,\fa\}$ and 
$\fa_b(h(b,t'))\not\in_{\con}\{\tr,\fa\}$. We show that
this is not possible: first notice
that because $b$ is not central in $\tr_b(g(b,t'))$,
a different atom must be central in this term, and this
atom must be $a$. For this to hold, $a$ must be central
in $\tr_b(t')$ and no atom different from $a$
can be tested by the first requirement of $\phi_{a,b,c}$.
So, after contraction of all further $a$'s we find
\[\tr_b(t')=_{\con}P\lef a \rig Q\]
with $P,Q\in\{\tr,\fa\}$, and similary 
$\fa_b(t')=_{\con}P'\lef c \rig Q'$
with $P',Q'\in\{\tr,\fa\}$.

If $P\not\equiv Q$ and $P'\not\equiv Q'$, 
then $t'$ is a 
term that satisfies $\phi_{a,b,c}$, but $t'$ 
is a term with lower $\#_{2p}$-value than 
$r(g(b,t')\lef b\rig h(b,t'))$, which is a 
contradiction.

Assume $P\equiv Q$ 
(the case $P'\equiv Q'$ is symmetric). 

Now $t=_{\con}r\big(\tr_b(g(b,t'))
\lef b \rig \fa_b(h(b,t'))\big)$, and 
no $b$'s can occur in
$T_b(g(b,t'))$, so
\[T_b(g(b,t'))\in_{\con}\{P\lef a\rig Q,~
\fa\lef (P\lef a\rig Q)\rig\tr, ~ 
(P\lef a\rig Q)\andthen \tr,~
(P\lef a\rig Q)\andthen \fa\}.\]
For $F_b(h(b,t'))$ a similar argument applies,
which implies that (recall $P\equiv Q$)
\[
\tr_b(g(b,t'))=_{\con}a\andthen P\text{ and }
\fa_b(h(b,t'))=_{\con}P'\lef c\rig Q'
\quad\text{with $P,P',Q'\in\{\tr,\fa\}$}.
\]
Assume $P\equiv \tr$ (the case $P\equiv\fa$ is symmetric).
So in this case
\[t=_{\con}r\big((a\andthen\tr)\lef b\rig(P'\lef c\rig Q')\big),\] 
and we distinguish two cases:

$i$. $P'\equiv \tr$ or $Q'\equiv\tr$. 
Now the reply to $a$ in $a\andthen \tr$ following
a positive reply to the initial $b$
has no effect, so this $a$ must be followed by another
central $a$.
But this last $a$ can also be reached
after a $b$ and a $c$, which contradicts 
$\phi_{a,b,c}$.

$ii$. $P'\equiv Q'\equiv\fa$. Since property 
$\phi_{a,b,c}$ holds it must be the case that
$a$ is a central condition in 
$r(\tr)$ with the property that $\tr_a(r(\tr))\ne_{\con}
\fa_a(r(\tr))$, otherwise the initial $b$ that stems 
from the substitution 
$x\mapsto (a\andthen\tr)\lef b\rig(c\andthen \fa)$ in $r(x)$
is upon reply \tr\ 
immediately followed by $a\andthen \tr$
and each occurrence of this $a$ is not able to yield 
both \tr\ and \fa, contradicting $\phi_{a,b,c}$. (And
also because this substitution yields no further 
occurrences of $b$ upon reply $\tr$.)

Similarly, $c$ is a central condition in $r(\fa)$ 
with the property that $\tr_c(r(\fa))\ne_{\con}
\fa_c(r(\fa))$.
We find that
$r(b)$
also satisfies $\phi_{a,b,c}$.
Now observe that $r(b)$  is a term with lower
$\#_{2p}$-value than 
$r(f(b,t'))$,
which is a contradiction.

\item We are left with four cases:
either $a$ is central in $\tr_b(g(b,t'))$ and 
$\fa_b(h(b,t'))\in_{\con}\{\tr,\fa\}$, or
$c$ is central in $\fa_b(h(b,t'))$ and 
$\tr_b(g(b,t'))\in_{\con}\{\tr,\fa\}$.
These cases are symmetric and it suffices to consider
only the first one, the others can be dealt with similarly. 

So assume $a$ is central in $\tr_b(g(b,t'))$ and 
$\fa_b(h(b,t'))=_{\con}\tr$. This implies
$\tr_b(g(b,t'))=_{\con}P\lef a\rig Q$
for some $P,Q$,
and after contraction of all 
$a$'s in $P$ and $Q$, 
\[\tr_b(g(b,t'))=_{\con}
P'\lef a\rig Q'\quad\text{ for some $P',Q'\in\{\tr,\fa\}$}.
\]
We find
\[t=_{\con}
r((P'\lef a\rig Q')\lef b\rig\tr),\]
and we distinguish two cases:

$i$. $P'\equiv\tr$ or $Q'\equiv\tr$. Now $c$ can be reached 
after a negative reply to $b$ according to $\phi_{a,b,c}$,
but this $c$ can also be reached
after a positive reply to $b$ and the appropriate reply 
to $a$, which contradicts 
$\phi_{a,b,c}$.

$ii$. $P'\equiv Q'\equiv\fa$. Since property 
$\phi_{a,b,c}$ holds it must be the case that
$a$ is a central condition in 
$r(\fa)$ with the property that $\tr_a(r(\fa))\ne_{\con}
\fa_a(r(\fa))$, otherwise the initial $b$ that stems 
from the substitution 
$x\mapsto (a\andthen\fa)\lef b\rig\tr$ in $r(x)$
is upon reply \tr\ 
immediately followed by $a\andthen \fa$
and each occurrence of this $a$ is not able to yield 
both \tr\ and \fa, contradicting $\phi_{a,b,c}$. (And
also because this substitution yields no further 
occurrences of $b$ upon reply $\tr$.)

Also, $c$ is a central condition in $r(\tr)$ 
with the property that $\tr_c(r(\tr))\ne_{\con}
\fa_c(r(\tr))$.
We find that
$r(b)$
also satisfies $\phi_{a,b,c}$.
Now observe that $r(b)$  is a term with lower
$\#_{2p}$-value than 
$r(f(b,t'))$,
which is a contradiction.

\end{enumerate}
This concludes our proof.
\end{proof}
\end{document}